\documentclass[11pt,a4paper,english]{article}
\usepackage[margin=2.5cm]{geometry}
\usepackage[T1]{fontenc}
\usepackage[utf8]{inputenc}
\usepackage{babel}
\usepackage[small]{caption}
\usepackage{graphicx}
\usepackage{amsmath}
\usepackage{amsthm}
\usepackage{booktabs}
\usepackage{algorithmicx}
\usepackage[ruled,linesnumbered]{algorithm2e}
\usepackage{bm}
\usepackage{tikz}
\usepackage{tikz-network}

\newtheorem{theorem}{Theorem}
\newtheorem{definition}{Definition}
\newtheorem{lemma}{Lemma}

\newtheorem{example}{Example}

\title{Combinatorial Procurement Auction in Social Networks}

\author{
Yuhang Guo\footnote{School of Computer Science and Engineering, University of Electronic Science and Technology of China.} 
\and  Dong Hao\footnote{School of Computer Science and Engineering, University of Electronic Science and Technology of China, Corresponding author, Contact Email: \texttt{haodongpost@gmail.com}}  \and  %
Bin Li\footnote{School of Computer Science, Nanjing University of Science and Technology.}
}
\date{}

\begin{document}
  \maketitle
  \begin{abstract}
This paper studies one emerging procurement auction scenario where the market is constructed over the social networks. In a social network composed of many agents, smartphones or computers, one requester releases her requirement for goods or tasks to suppliers, then suppliers who have entered the market are also encouraged to invite some other suppliers to join and all the suppliers in the network could compete for the business. The key problem for this networked auction is about how to incentivize each node who have entered the sell not only to truthfully use her full ability, but also to forward the task to her neighbours. Auctions conducting over social networks have attracted considerable interests in recent years. However, most of the existing works focus on classic forward auctions. Moreover, there is no existing valid networked auction considering multiple goods/tasks. This work is the first to explore procurement auction for both homogeneous and heterogeneous goods or tasks in social networks.  From both theoretical proof and experimental simulation, we proved that the proposed mechanisms are proved to be individual-rational and incentive-compatible, also both the cost of the system and the requester could get decreased.
    \vspace{0.4cm}
    
  \end{abstract}

  \newpage

\section{Introduction}

Auction is a fundamental paradigm for realizing desirable resource allocation, which has been applied in various realistic scenarios. Classic auctions may focus on conducting desired social choices in a fixed group of participants who can be directly informed by the organizer. Conventionally, they never take the underlying social interactions between participants into account. However, in most human society and computer systems, network plays a central role. Very recently, auction design in social networks has been a promising direction \cite{ijcai2021-605}, and this research field novely and finely integrates the two fundamental disciplines: mechanism design and networks. By introducing networks and social interactions between agents into traditional auction theory, the main difficulty lies in the conflict that from one aspect, the organizer wishes to attract more people to join the auction in order to increase her revenue, from the other side, participants would be unwilling to bring more competitors into the market, which only makes herself more difficult to win the auction. So networked auction mechanisms concern about how to incentivize participants not only to truthfully bid but also to invite their friends, which can be used to simultaneously improve both the social welfare and the seller's revenue. Past few years have witnessed considerable progress to tackle the underlying contradiction between system optimality and the individuals' self interests in the networked auction area.

 For networked auction with one item for sale, the pioneer work was proposed in  \cite{li2017mechanism}, which proved that the
 state-of-art VCG mechanism could incur the organizer terrible deficit in extreme network cases and further formulated the novel IDM mechanism. IDM is proved to be the first work to satisfy incentive-compatible while guaranteeing non-negative revenue for the organizer, leading the way for follow-up works. Later, various single item diffusion auction mechanisms  \cite{li2019diffusion,zhangECAIincentivize,zhang2020redistribution,LI2022103631,jeong2020groupwise} focusing on different attributes and different scenarios have also been proposed. CDM proposed by  \cite{li2019diffusion} formulates one mechanism family satisfying great properties where IDM could be considered as one special case. For single-item networked auctions with transferring cost, CSM \cite{li2018customer} and WDM \cite{li2019diffusion} give solutions in this domain. Moreover, FDM \cite{zhangECAIincentivize} and NRM \cite{zhang2020redistribution} study redistribution mechanisms issues over social networks, where more participants are rewarded for propagating auction information, which could help greatly in the expansion of the market in practice.

Although there has been emerging interests in diffusion auction design, most of the works focus on single-item auctions over social networks, this paper is the first attempt to extend diffusion auctions into a setting of combinatorial items. The study of traditional combinatorial auctions has received much academic attention  \cite{cramton2004combinatorial}. Traditional combinatorial auctions study the market with a large number of items. Buyers in combinatorial auctions are allowed to express their own preferences on different number or bundle of items. Spectrum auctions  \cite{cramton1997fcc,cramton2000collusive} is one of the most famous applications in this area, by which a large number of licenses have been successfully sold. Another common application of combinatorial auction is the procurement market  \cite{cramton2004combinatorial}. The purchase or procurement by governments, corporations, social organizations or individuals have always been an important part of market trading. Reverse auction or procurement auction is one of the most common method these entities take for procuring goods or tasks.

For example of reverse auctions, consider the procurement in transportation issues  \cite{bichler2020strategyproof}: one corporation wants to establish logistic routes between different sites. Doing this work by herself is very costly, thus a better way is to procure transportation services from different suppliers who can provide routes between these sites. For each supplier, the cost of providing a bundle of routes could be different due to the structure of demanded route. So solving such a complex combinatorial auction problem for corporations could help them save large amount of money. Another scenario is called crowdsensing  \cite{zhang2015truthful,yang2015incentive}, which has become more popular with the advance of Internet technology and smart phones. In crowdsensing, social organizations or individuals can broadcast their procurement requirements for some goods or tasks via social media, then people nearby who receive the procurement message can provide the corresponding tasks and charge some fees to the crowdsensing launcher. 

In this paper, we are going to focus on two forms of combinatorial reverse auctions (CRA), one is multi-unit reverse auctions with additive valuation bidders and the other is heterogeneous goods procurement with single-minded bidders. When extending single item diffusion auction into multiple items, even for single demand bidders, designing truthful mechanisms become rather complicated  \cite{zhao2021,ijcai2021-605}. 
To date, Zhao \textit{et al.} \cite{zhao2019sellmultipleitems} and Kawasaki \textit{et al.}  \cite{kawasaki2020strategy} are the only two existing works that try to solve this task. 
From the perspective of buyer-centric auction models, Liu \textit{et al.} \cite{Liu_Wu_Li_Wang_2021} firstly studied procurement auctions via social networks under certain budget with logarithmic approximation procured value. 

When extending diffusion auction into combinatorial settings, this task becomes complicated since the bidding is tightly coupled with each participant's diffusion strategy.
In our paper, we point out the drawbacks of the existing works on multiple item diffusion auctions, and analyse that the core problems lies in designing monotone allocation rules. We reveal that, to design an incentive compatible multi-unit auction or even combinatorial auction, a monotone allocation is necessary.  
To tackle the allocation monotonicity for auctions with combinatorial items, we divide the networked market into independent sub-markets according to the graph structure. We then design one allocation rule which is both valuation-monotone and diffusion-monotone. After that, we derive the payment of each agent as her critical winning bid, which finally leads to a truthful mechanism. In addition, we also prove our mechanisms have some other good properties.

Section 2 introduces how we model combinatorial reverse auctions in social networks. Section 3 analyses why in multi-unit diffusion auctions, a non-monotone allocation could bring big troubles for truthfulness. Then in section 4 and 5, we propose novel diffusion auction mechanisms.

\section{Networked Reverse Auction Model}
Consider a social network $G=(V,E)$. The node set $V=N\cup \{p\}$ represents all agents in a market where there is a single requester $p$ and other nodes are suppliers $N$. Set $E$ contains the directed edges between any two agents, and an directed edge represents that the agent on the tail can contact the one on the head. 

For the requester $p$, denote her type file as $t_p=(\boldsymbol{\tau},\bar{\boldsymbol{v}})$, where $\boldsymbol{\tau}$ is the set of her procurement tasks while $\bar{\boldsymbol{v}}$ represents the set of her reserve costs for finishing each task by herself. The requester can either finish the tasks by herself or outsource them to other agents. The reserve cost $\bar{\boldsymbol{v}}$ is just the cost that the requester finishes these tasks by herself. Note that the mechanism and theory in this paper apply to any social network and the requester can be an arbitrary node. 

For any supplier $i\in N$, denote her profile as $t_i=(a_i,c_i,r_i)$ where $a_i$ is the tasks that she can finish, $c_i$ is her (unit) supply cost and $r_i$ contains those suppliers who $i$ can inform about the procurement launched by requester $p$. Let $\boldsymbol t$ be the type profile from all suppliers and let ${\boldsymbol c}_N$ be the profile of costs of all suppliers. A supplier can contact other suppliers and ask them to join in the procurement.
We consider the scenario that $a_i$ is public knowledge, while $c_i,r_i$ are private information which could be manipulated by $i$ herself. We also assume that the requester herself who can finish the posted procurement goods or tasks with her own reserved cost, but she wants to reduce the expenditure she crowdsources the tasks into the social network. We call the requester herself the \textit{virtual supplier}, and denote it as $\phi$ where $t_{\phi}=(\boldsymbol{\tau},\bar{\boldsymbol{v}},\emptyset)$.  

For scenarios as described above, we say the requester and the suppliers engage in a networked reverse auction.
A networked reverse auction mechanism (abbr., RAN) $\mathcal{M}=(\boldsymbol{\pi},\boldsymbol{x})$ consists of an allocation policy $\boldsymbol{\pi}=\{\pi_i\}_{i\in N}$ and a payment rule $\boldsymbol{x}=\{x_i\}_{i\in N}$. According to the allocation and payment, for each supplier $i$, define her utility as $u_i(\boldsymbol{t},\mathcal{M})=x_i - c_i$. For any supplier $i$, define $d(p,i)$ as the \textbf{shortest distance} from the requester $p$ to the supplier $i$. To characterize networked reverse auctions, the following basic definitions are needed. 



\begin{definition} 
A networked reverse auction $\mathcal{M}=(\boldsymbol{\pi},\boldsymbol{x})$ is individual rational if  $\forall i\in N$, $u_i(t_i,\boldsymbol{t}_{-i}^\prime)\geq 0$.
\end{definition}

\begin{definition} 
A networked reverse auction $\mathcal{M}=(\boldsymbol{\pi},\boldsymbol{x})$ is incentive compatible if $\forall i\in N$, $u_i(t_i,\boldsymbol{t}_{-i}^\prime)\geq u_i(t_i^\prime,\boldsymbol{t}_{-i}'')$.
\end{definition}

Note that on the right side of the inequality, $(t_i,\boldsymbol{t}^\prime_{-i})$ is replaced by $(t_i^\prime,\boldsymbol{t}^{\prime\prime}_{-i})$ since when $i$ changes her type, such as not inviting all of her neighbors, then some suppliers who can enter the market under $t_i$ may no longer receive the procurement information. Thus, when under $t_i^\prime$, we would change the type profile of other suppliers except $i$ from $\boldsymbol{t}^\prime_i$ to $\boldsymbol{t}^{\prime\prime}_i$.

Individual rationality (IR) guarantees that any supplier who truthfully reports her type file will never suffer from deficit; Incentive compatibility (IC) represents that for all suppliers, truthfully submitting her type (i.e., both truthfully bidding and inviting \textit{all} her friends) dominates strategically reporting any other type in her strategy space. Furthermore, we have the following notion about the requester's cost. 
\begin{definition}
A networked reverse auction $\mathcal{M}=(\boldsymbol{\pi},\boldsymbol{x})$ is weakly budget balanced if the total cost for the requester to hire some suppliers never surpass her total reversed cost. That is, $u_p=\sum_{\tau_i\in \boldsymbol{\tau}}\bar{v}(\tau_i)-\sum_{i\in N}x_i\geq 0$.
\end{definition}

In a networked reverse auction, for any two different types $t_i^1=(a_i^1,c_i^1,r_i^1)$ and $t_i^2=(a_i^2,c_i^2,r_i^2)$, we define the \textbf{partial ordering} over types $t_i^1\succ t_i^2$ \textit{iff} $c_i^1\leq c_i^2\wedge r_i^1\subseteq r_i^2$. This partial ordering portrays that for any supplier, reporting a lower cost and diffusing information to fewer friends is more preferable, because it makes an agent more competitive in a crowd  \cite{li2020incentive}. Taking into consideration of this partial ordering, the allocation monotonicity is defined as follows.

\begin{definition}\label{allocation-monotonic}
An allocation policy $\boldsymbol{\pi}$ is monotone if $\forall t_i^1\succ t_i^2,\boldsymbol{\pi}(t_i^1,\boldsymbol{t}_{-i})\geq \boldsymbol{\pi}(t_i^2,\boldsymbol{t}_{-i})$.
\end{definition}
By this definition, the monotonicity of allocation rule should not only cover the valuation, but also consider the diffusion action space. Lower cost with inviting fewer friends leads to a higher probability to win. For this overall allocation monotonicity, we decompose it with respect to the value domain and the diffusion domains in definition \ref{value-monotone} and definition \ref{diffusion-monotone}, corresponding to the agents' report cost the the invitation set. 

\begin{definition}[Value-Monotone Allocation]\label{value-monotone}
An allocation policy in networked reverse auction $\boldsymbol{\pi}$ is value-monotone if $\forall (a_i^1,c_i^1,r_i^1)\succ (a_i^1,c_i^2,r_i^1)$ where $c_i^1 \leq c_i^2$, $\boldsymbol{\pi}((a_i^1,c_i^1,r_i^1),\boldsymbol{t}_{-i}) \geq \boldsymbol{\pi}((a_i^1,c_i^2,r_i^1),\boldsymbol{t}_{-i}) $.
\end{definition}

\begin{definition}[Diffusion-Monotone Allocation]\label{diffusion-monotone}
An allocation policy in networked reverse auction $\boldsymbol{\pi}$ is diffusion-monotone if $\forall (a_i^1,c_i^1,r_i^1)\succ (a_i^1,c_i^1,r_i^2)$ where $r_i^1 \subseteq r_i^2$, $\boldsymbol{\pi}((a_i^1,c_i^1,r_i^1),\boldsymbol{t}_{-i}) \geq \boldsymbol{\pi}((a_i^1,c_i^1,r_i^2),\boldsymbol{t}_{-i}) $.
\end{definition}
 
Intuitively, these two definitions about monotonicity reflect that, for each agent, when fixed the invitation domain, reporting lower working cost would not result in worse allocation results. On the other hand, when fixed the cost reporting domain, inviting fewer friends would not result in worse allocation results.


Besides the monotonicity, we here define that an allocation is \textbf{feasible} if under this allocation, all procurement requirement get satisfied either from suppliers or from requester herself. That is, an allocation policy $\boldsymbol{\pi}$ is \textit{feasible} iff $\cup_{i\in N\cup\{\phi\}}\pi_i=\boldsymbol{\tau}$. Moreover, we define optimal allocation as follows.

\begin{definition}\label{definition:optimal_allocation}
For networked reverse auction mechanisms, let $\Pi$ be the set of feasible allocation policies. A feasible allocation $\boldsymbol{\pi}^\ast$ is \textbf{optimal} if it minimizes the total social cost, i.e.,
$\boldsymbol{\pi}^\ast=\arg\min_{\boldsymbol{\pi}^\prime\in \Pi } \boldsymbol{\pi}^\prime(\boldsymbol{t}) \cdot {\boldsymbol{c}}_N$, where $\boldsymbol{c}_N=(c_1,\cdots,c_N)$.  
\end{definition}

\begin{definition}
Given an allocation rule $\boldsymbol{\pi}$ and all other suppliers' report type $\mathbf{t}_{-i}^\prime$, supplier $i$'s \textbf{critical winning cost} with $r_i^\prime\subseteq r_i$ is $c_i^\ast=\arg\max_{c_i^\prime}\boldsymbol{\pi}_i((c_i^\prime,r_i^\prime),\boldsymbol{t}_{-i}^\prime) = 1$.
\end{definition}

For any supplier $i$, when fixed all other suppliers' type $\boldsymbol{t}^\prime_{-i}$ and $i$'s diffusion dimension $r_i^\prime$, $i$'s \textbf{critical winning cost} $c_i^\ast$ is the maximum cost that she can report to be selected for the business.

\begin{definition}
For each supplier $i$, she could decide some other suppliers' participation by taking different diffusion strategies. We denote those suppliers who are dominated by supplier $i$ as $i$'s children by $\mathcal{D}_i$ where $\mathcal{D}_i=N(r_i)\setminus N(\emptyset)$. Note $N(r_i)$ represents the suppliers set when $i$ invites all of her friends while $N(\emptyset)$ represents the suppliers set when $i$ invites no friend. 
\end{definition}

\section{VCG and Diffusion Auction Mechanisms}
In this section, firstly, we illustrate that the VCG mechanism can be applied to our scenario but has a fatal flaw which makes it impossible to implement in reality. Next, we would explore that there exists great gap between single item diffusion auctions and multiple ones. Further, we clarify that non-monotone allocation rules of some multi-unit diffusion auction mechanisms will make the payment rule design intractable. Then we propose that the key issue for designing diffusion forward or reverse auction mechanisms is to find one monotone allocation rule and identify the corresponding payment according to the critical winning bid for each agent. 

\subsection{The Dilemma of Multi-unit Networked Auctions}
When conducting conventional forward auction scenarios over social networks,  Li \textit{et al.} \cite{li2017mechanism} has pointed that the state-of-art VCG mechanism would encounter terrible deficit over some specific social networks in single item diffusion auction, i.e. the seller could sell the item with high social welfare but get negative revenue, which is unacceptable in reality. When it comes to our procurement scenarios with multiple items, this phenomenon is further worsened. We would further mention more details in the simulation section. 

Note that we define the allocation monotonicity over diffusion auctions in definition \ref{value-monotone} and definition \ref{diffusion-monotone}. Intuitively, value-monotonicity requires that for any agent $i$, when fixing $\boldsymbol{t}_{-i}$ and $r_i$, reporting lower cost (higher bid in forward diffusion auctions) would never make worse allocation; diffusion-monotoncity requires that for any supplier $i$, when fixing $\boldsymbol{t}_{-i}$ and $c_i$, inviting fewer suppliers (In both forward/reversed diffusion auctions) would never make worse allocation. 

For single-item diffusion auctions, the pioneer work IDM (Information Diffusion Mechanism) introduced in \cite{li2017mechanism} gives one direction, Algorithm.\ref{idm} shows how this mechanism works. 

{\small{\begin{algorithm}
\SetAlgoLined
\SetKwInOut{Input}{Input}\SetKwInOut{Output}{Output}
\Input{$G=(N\cup\{s\},E),\boldsymbol{t}=(t_1,\cdots,t_N),\forall i\in N,t_i=(b_i,r_i)$}
\Output{$\boldsymbol{\pi},\boldsymbol{x}$}
Locate the highest bid agent $m$ in market $G$\;
Calculate the critical diffusion sequence $\boldsymbol{\ell}=(1,2,\cdots,m)$\;
\For{$i \in \boldsymbol{\ell}$}
{
\eIf{$b_i=\max\{b_j\},\forall j\in \boldsymbol{t}_{-(i+1)}$}
{
$\boldsymbol{\pi}_i\leftarrow 1$; $x_i\leftarrow b^\ast_{-i}$\;
\textbf{Break}\;
}
{
$\boldsymbol{\pi}_i\leftarrow 0$; $x_i\leftarrow b^\ast_{-i}-b^
\ast_{-(i+1)}$
}
}
\caption{Information Diffusion Mechanism in Single-item Diffusion Auction}
\label{idm}
\end{algorithm}}}

Intuitively, IDM makes allocation along the \textit{critical diffusion sequence}\footnote{See \cite{li2017mechanism} for detailed definition of critical diffusion sequence.} and allocate the item to the first agent whose bid is highest when not diffusing to the next on-sequence agent. The winner pays the highest bid without her self's participation while other on-sequence agents could get some reward. It is not hard to verify that the allocation rule satisfies both value-monotonicity and diffusion-monotonicity. However, when applying the IDM to multi-unit diffusion auctions, even in unit-demand settings, the monotonicity gets broken. Here the toy example gives inspiration. 

\begin{figure}[!htbp]
    \centering
\scalebox{0.8}{
\begin{tikzpicture}
    \Vertex[color=yellow,size=1, x=-0.5, y=0,fontscale=1.4, label=$s$]{p}
    \Vertex[color=white, x=1, y=1, size=1, label=$3$,fontscale=1.3
    ]{a}
    \Vertex[color=white, x=1, y=-1, size=1, label=$1$, fontscale=1.3
   ]{b}
    \Vertex[color=white, x=2.8, y=-1, size=1, label=$2$, fontscale=1.3
    ]{c}
    \Vertex[color=white, x=4.6, y=-1.8, size=1, label=$5$,fontscale=1.3 ]{d}
    \Vertex[color=white, x=4.6, y=0, size=1, label=$10$,fontscale=1.3,RGB,color={127,201,127} ]{e}
    \Vertex[color=white, x=2.8, y=1,size=1, label=$9$,fontscale=1.3]{f}
    \Edge[lw=1,Direct](p)(a)
    \Edge[lw=1,Direct](p)(b)
    \Edge[lw=1,Direct](b)(c)
    \Edge[lw=1,Direct](c)(d)
    \Edge[lw=1,Direct](c)(e)
    \Edge[lw=1,Direct](a)(f)
    \Text[x=-3.3,y=0,fontsize=\large,color=red,style={draw,rectangle}]{Allocate $2$ unit goods}
    \Text[x=1,y=0.3]{$a$}
    \Text[x=1,y=-0.3]{$b$}
    \Text[x=2.8,y=-0.3]{$c$}
    \Text[x=2.8,y=0.3]{$d$}
    \Text[x=5.3,y=0]{$e$}
    \Text[x=5.3,y=-2]{$f$}
\end{tikzpicture}
}
    \caption{Multi-unit diffusion auction with single-unit demand bidders: $s$ represents the seller who has two unit item for sale, there exist $6$ agents $\{a,b,c,d,e,f\}$ under the market.}
    \label{idm_case}
\end{figure}
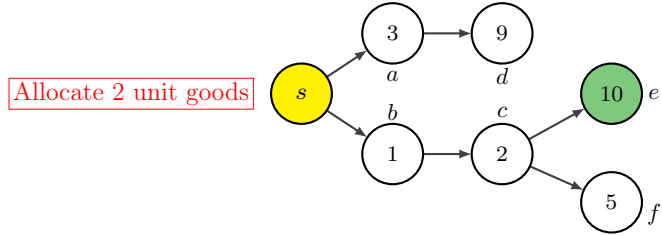

Considering applying IDM under this scenario, the most direct way is to allocate items one by one, firstly, we

\subsection{Counterexample for Existing Diffusion Auction Mechanisms}
In Alg.\ref{sm1}, we devise a multi-unit reverse auction mechanism for information diffusion scenario. This mechanism instructs that, following the ranking of $d(p,i)$, for each node $i$, her allocation decides on the optimal allocation in Definition \ref{definition:optimal_allocation} when $i$'s children are removed (children is defined as $\mathcal{D}_i=N(r_i)\setminus N(\emptyset)$ (line 6 in Alg.\ref{sm1}), i.e, those agents who can not enter this market without $i$'s invitation). Her payment equals the social cost increment due to her absence (also with her children) and only her children's absence. Intuitively, this mechanism can be seen as one extension on VCG mechanism. However, we implement this algorithm into the following counter example and explain that its allocation violates monontonicity in Definition \ref{allocation-monotonic} and this mechanism is not IC.

{\small{\begin{algorithm}
\SetAlgoLined
\SetKwInOut{Input}{Input}\SetKwInOut{Output}{Output}
\Input{$G=(V,E),\boldsymbol{t}_p=(\tau,\bar{v}),\boldsymbol{t}=(t_1,\cdots,t_N)$}
\Output{$\boldsymbol{\pi},\boldsymbol{x}$}
Reorder agents by distance $d(p,1)\prec \cdots\prec d(p,N)$\;
\For{$i\leftarrow 1$ \KwTo $N$}
{
\If{$\tau\leq 0$}{\textbf{Break}}

$\mathcal{D}_i\leftarrow N(r_i)\setminus N(\emptyset)$\;
$\boldsymbol{\pi}^\ast({\boldsymbol{t}_{-\mathcal{D}_i}})\leftarrow\arg\min _{\boldsymbol{\pi}^\prime}\boldsymbol{\pi}^\prime(\boldsymbol{t}_{-D_i}) \boldsymbol{c}_{-\mathcal{D}_i}$\;

$\boldsymbol{\pi}^\ast({\boldsymbol{t}_{-i}})\leftarrow\arg\min _{\boldsymbol{\pi}^\prime}\boldsymbol{\pi}^\prime(\boldsymbol{t}_{-i}) \boldsymbol{c}_{-i}$\;

$\pi_i\leftarrow \boldsymbol{\pi}^\ast_i({\boldsymbol{t}_{-\mathcal{D}_i}})$,$\tau \leftarrow \tau - \pi_i$\;

$x_i\leftarrow \boldsymbol{\pi}^\ast(\boldsymbol{t}_{-i})\boldsymbol{c}_{{-i}} - (\boldsymbol{\pi}^\ast(\boldsymbol{t}_{-\mathcal{D}_i})\boldsymbol{c}_{{-\mathcal{D}_i}}-\pi_i c_i)$\;

}
\caption{A Non-Monotone Diffusion Auction}
\label{sm1}
\end{algorithm}}}


\begin{example}\label{exm1}

In Fig.\ref{example-1}, the circle node represents the requester $p$ and each oval node represents one supplier. In each oval node, the number on the left is the supplier's unit cost $c_i$, while the number on the right represents number of tasks she can finish. the requester $p$ has a procurement task of $4$ units of goods. Initially, she can only invites supplier $a$ and $b$. Then the market can be expanded due to that $b,c,d$ can spontaneous diffuse the procurement information. Implementing Alg.\ref{sm1}, green oval nodes are the selected suppliers.  
\end{example}
\begin{figure}[h]
	    \centering
    	  \includegraphics[scale=1]{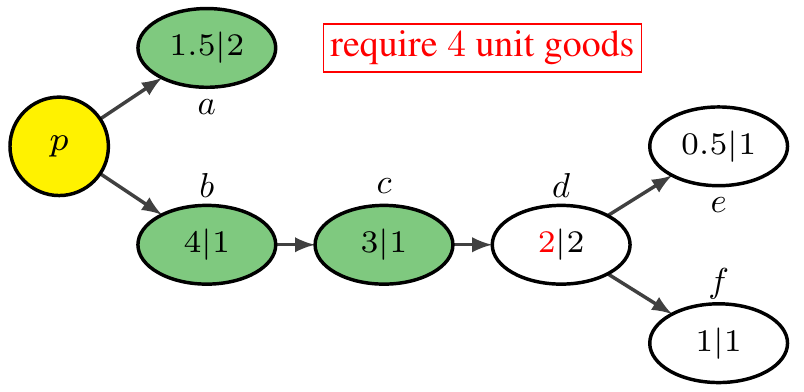}
    \caption{A counter example where the non-monotone diffusion auction mechanism in Alg.\ref{sm1} fails to be IC.}
    \label{example-1}
\end{figure}
However, Example \ref{example-1} can be used to prove that the allocation in Alg.\ref{sm1} is not monotone and this non-monotonicity makes this mechanism not IC. 
%
%
To see this, consider node $d$. If $d$ does not diffuse the procurement information to $e$ and $f$, then according to the allocation rule in Alg.\ref{sm1}, whatever cost $d$ misreports nodes $\{a,b,c\}$ will always be selected. However, if she does diffuse, then when the allocation rule examines $a$, it will be found that, as long as $d$ misreports to a unit cost less than $1.5$, $d$'s diffusion can exclude agent $a$ and make herself become selected, whereby her payment becomes larger than that when not diffusing. This indicates this non-monotone allocation leads to a mechanism that is not IC.

The mechanisms given by   \cite{zhao2019sellmultipleitems} and   \cite{kawasaki2020strategy} are only existing efforts for multi-item forward diffusion auction scenario. Unfortunately, the allocation rules of these two mechanisms also fail to be monotone, which finally makes these two existing mechanism be not incentive compatible. Here we give the specific counterexample that they fail to be incentive compatible. 

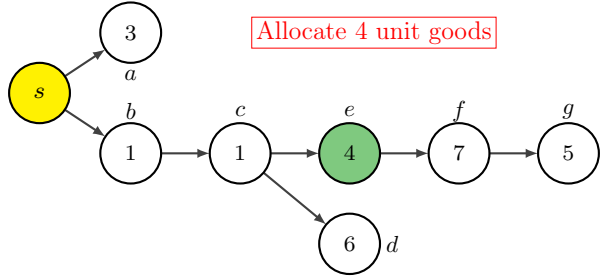
\begin{figure}[!htbp]
    \centering
\scalebox{0.8}{
\begin{tikzpicture}
    \Vertex[color=yellow,size=1, x=-0.5, y=0,fontscale=1.4, label=$s$]{p}
    \Vertex[color=white, x=1, y=1, size=1, label=$3$,fontscale=1.3
    ]{a}
    \Vertex[color=white, x=1, y=-1, size=1, label=$1$, fontscale=1.3
   ]{b}
    \Vertex[color=white, x=2.8, y=-1, size=1, label=$1$, fontscale=1.3
    ]{c}
    \Vertex[color=white, x=4.6, y=-1, size=1, label=$4$, fontscale=1.3,RGB,color={127,201,127}]{d}
    \Vertex[color=white, x=6.4, y=-1, size=1, label=$7$,fontscale=1.3]{e}
    \Vertex[color=white, x=4.6, y=-2.5, size=1,label=$6$,fontscale=1.3]{f}
    \Vertex[color=white, x=8.2, y=-1, size=1, label=$5$,fontscale=1.3]{g}
    
    \Edge[lw=1,Direct](p)(a)
    \Edge[lw=1,Direct](p)(b)
    \Edge[lw=1,Direct](b)(c)
    \Edge[lw=1,Direct](c)(d)
    \Edge[lw=1,Direct](d)(e)
    \Edge[lw=1,Direct](c)(f)
    \Edge[lw=1,Direct](e)(g)
    \Text[x=5,y=1,fontsize=\large,color=red,style={draw,rectangle}]{Allocate $4$ unit goods}
    \Text[x=1,y=0.3]{$a$}
    \Text[x=1,y=-0.3]{$b$}
    \Text[x=2.8,y=-0.3]{$c$}
    \Text[x=4.6,y=-0.3]{$e$}
    \Text[x=6.4,y=-0.3]{$f$}
    \Text[x=8.2,y=-0.3]{$g$}
    \Text[x=5.3,y=-2.5]{$d$}
\end{tikzpicture}
}
    \caption{Counterexample for GIDM and DNA-MU: the yellow circle represents the seller who has $4$ items for sale. Initially, she invites $a,b$ to join the market, then $b$ invites $c$, etc. Finally, this networked market gets formed. Number in each white circle is the valuation of each bidder.}
    \label{fig1}
\end{figure}

This example was firstly proposed in   \cite{takanashi2019efficiency} to prove the failure of truthfulness in GIDM. However, we carefully implement the mechanism under DNA-MU and find the same problem for truthfulness. Firstly, we describe the procedures of these two mechanisms briefly and point out where they violate the truthfulness property. For GIDM, running procedures can be checked in Section.4 in   \cite{takanashi2019efficiency}. So we will not go into too much details and directly we give the final allocation: winners set is: $\{b,c,d,e\}$. Next, we give the running process of DNA-MU (All of our notations are consistent with that in   \cite{kawasaki2020strategy}): 

\begin{itemize}
\item[1.]
  Determine the priority order by distance:
  \(a(1) \succ b(1) \succ c(2) \succ d(2) \succ e(2) \succ f(3) \succ g(4)\).
\item[2.]
  Initialize winner set as \(W=\emptyset\) and then decide bidders' allocation
  and payment:

  For agent \(a\),
  \(p_a=v^\ast(\hat{N}_{-a}\setminus W,k^\prime)=v^\ast(\hat{N}_{-a}\setminus\emptyset,4)=4 > 3\).
  \(f_a=0\).

  For agent \(b\),
  \(p_b=v^\ast(\hat{N}_{-b}\setminus W,k^\prime)=v^\ast(\hat{N}_{-b}\setminus \emptyset,4)=0 < 1\).
  \(f_b=1,W\leftarrow W\cup\{b\}, k^\prime \leftarrow k^\prime -1\).

  For agent \(c\),
  \(p_c=v^\ast(\hat{N}_{-c}\setminus W,k^\prime)=v^\ast(\hat{N}_{-c}\setminus \{b\},3)=0 < 1\).
  \(f_c=1,W\leftarrow W\cup\{c\},k^\prime \leftarrow k^\prime -1\).

  For agent \(d\),
  \(p_d=v^\ast(\hat{N}_{-d}\setminus W,k^\prime)=v^\ast(\hat{N}_{-d}\setminus \{b,c\},2)=5 < 6\).
  \(f_d=1,W\leftarrow W\cup\{d\},k^\prime \leftarrow k^\prime -1\).

  For agent \(e\),
  \(p_e=v^\ast(\hat{N}_{-e}\setminus W,k^\prime)=v^\ast(\hat{N}_{-e}\setminus \{b,c,d\},1)=3 < 4\).
  \(f_e=1,W\leftarrow W\cup\{e\},k^\prime \leftarrow k^\prime -1\).

  \(k\) decremented to \(0\) and the price for agent \(f\) and \(g\)
  become infinity.
 \item[3.] Final winners are $W=\{b,c,d,e\}$.
\end{itemize}

This two mechanisms output the same allocation under bidders truthfully reporting, however, when $f$ misreports her $r_f^\prime$ from $\{g\}$ to $\emptyset$. Then the allocation changes into $\{a,b,c,f\}$ where $f$ gains extra payoff by changing herself from loser to winner. Thus, for this two mechanisms, they are both not truthful. 

Further, we could check that both of the two mechanisms' allocation rules are not monotone. Note the bidder $e$ marked in green: when reporting $r_e^\prime = r_e$, she would be a winner as if reporting $v_e^\prime \geq 3$; when reporting $r_e^\prime =\emptyset$, she would be a winner as if reporting $v_e^\prime \geq 6$. So here we can see that inviting fewer friends makes agent $e$ more difficult to be one winner, which contradict the allocation monotonicity: reporting higher bid and diffusing information to fewer friends is more preferable to be a winner. Because of the non-monotonicity of agent $e$, her children agent $f$ has the possibility to gain extra utility by manipulating her own reporting type.

\subsection{Why Non-Monotone Allocation Fails}
The underlying reason why it is hard to ensure IC in multi-item diffusion auctions is that (1) to determine the allocation rule, the social network needs to be transformed into a tree structure capturing the information flow, where there are many branches and each node may have her upstream nodes and downstream nodes in a branch. (2) If a focal agent has very competitive downstream agents on her branch, this focal agent can take advantage of these competitive child nodes, which can deprive the allocation of agents on other branches, and transfer the allocation to the focal player's own branch. 
(3) After the allocation transferring to the focal player's branch, any node on this branch (including the focal agent herself) can strategize her bid to gain extra utility. Any node (including the focal agent herself) can also manipulate her type to change the allocation of her parent node, etc.


Under single parameter environment, Myerson's Lemma   \cite{myerson1981optimal} gives the very concise instructions for designing IC mechanism: valuation-monotone allocation rule with each agent pays the critical winning bid. 
Specifically, by introducing diffusion dimension, networked auctions turn into a non-typical multi-dimension environment.
For such generalized auction scenarios with only one item, the allocation monotonicity was firstly identified in   \cite{li2020incentive} and is proved to be a \textit{sufficient} condition for any single item diffusion auction to be IC.  
Therefore, if an allocation rule is non-monotone, 
to guarantee truthfulness, it is necessary to find critical winning bid for each agent, which means one needs to check all range of bid and every subset of her neighbour set, which is computationally intractable when the market scale expands. 
In   \cite{li2020incentive}, based on the monotonicity, they formulated the optimal corresponding payment rule: winner pays the critical winning bid when inviting no friends, loser pays the difference between critical winning bid when inviting no friends and critical winning bid when inviting all friends. So, one promising direction for designing IC diffusion auction mechanisms for multi-item diffusion forward/reversed auctions probably lies in finding one monotone allocation rule and searching for critical winning bid to decide payment.


\section{Networked Procurement Auction for Homogeneous Goods}
In this section, we investigate the combinatorial reverse auction for homogeneous goods under social networks. $p$ posts her required $\tau$($\vert\boldsymbol{\tau}\vert$) homogeneous items $\boldsymbol{\tau}$. Denote her private reserve cost for each unit as $\bar{v}$. Each supplier reports a number $a_i$ of items she can provide, unit cost $c_i$ and she can diffuse the procurement information to neighbours in $r_i$. 

\subsection{Market Division}
For multi-dimensional diffusion auction, the most direct way is to apply the VCG mechansim
  \cite{vickrey1961counterspeculation,clarke1971multipart,groves1973incentives}. However, it is proved that even in single-item forward diffusion auction, VCG is not weakly budget balanced   \cite{li2017mechanism}. Finding a monotone allocation in diffusion auctions is not easy. One operational way is to divide the network into different layers, and in each layer, secure the monotone allocation.
See Alg.\ref{nmd}.

{\small{
\begin{algorithm}
\SetAlgoLined
\SetKwInOut{Input}{Input}\SetKwInOut{Output}{Output}
\SetKwFunction{ShortestPathAlgorithm}{ShortestPathAlgorithm}
\Input{Social network $G=(V,E)$, $N=|V|$}
\Output{Market division $\mathcal{G}=(\mathcal{G}_1,\cdots,\mathcal{G}_{d^\ast})$}
\For{$i\leftarrow 1\text{ to } N$}
{
$d(p,i)\leftarrow$ \ShortestPathAlgorithm{p,i}\;
}
Initialize sub-market $\mathcal{G}_{d(p,i)}=\emptyset,\forall d(p,i)$\;
Compute $d^\ast=\max\{d(p,i),\forall i\in N\}$\;
\For{$i\leftarrow 1\text{ to } N$}
{
$\mathcal{G}_{d(p,i)}\leftarrow \mathcal{G}_{d(p,i)} \cup \{i\}$\;
}
\caption{Networked Market Division}
\label{nmd}
\end{algorithm}}}

Alg.\ref{nmd} first computes the shortest distance from each supplier to the requester (any shortest path algorithm works), then all suppliers are categorised into different market divisions according to its distance to the requester (line 4-7). We have the following basic observation about market divisions.

\begin{lemma}\label{lem1}
For any sub-market $\mathcal{G}_{i}$, $\forall j\in \mathcal{G}_i$, misreporting $r_j^\prime \subseteq r_j$ will never change the sub-market $\mathcal{G}_i$.
\end{lemma}
\begin{proof}
If there are two suppliers $s_1$ and $s_2$ in $\mathcal{G}_i$, obviously $d(p,s_1)=d(p,s_2)$. Assume that $s_1$ can change $d(p,s_2)$ by mireporting $r_{s_1}^\prime$, then $s_1$ must be on the shortest path from $p$ to $s_2$, i.e. $d(p,s_2)\geq  d(p,s_1)+1$ which contradicts $d(p,s_1)=d(p,s_2)$. Moreover, nobody could misreport to change her own shortest distance to $p$. Thus, for any supplier $s$ in $\mathcal{G}_{i}$, misreporting $r_s$ does not change the market she is in.
\end{proof}

{\small{
\begin{algorithm}[h]
\SetAlgoLined
\SetKwInOut{Input}{Input}\SetKwInOut{Output}{Output}
\SetKwFunction{NetworkedMarketDivision}{NetworkedMarketDivision}
\Input{$G=(V,E),\boldsymbol{t}_p=(\boldsymbol{\tau},\bar{\boldsymbol{v}}),\boldsymbol{t}=(t_1,\cdots,t_N)$}
\Output{$\boldsymbol{\pi},\boldsymbol{x}$}
$(\mathcal{G}_1,\cdots,\mathcal{G}_{d^\ast})\leftarrow$ \NetworkedMarketDivision{$G$}\;
\textbf{Initialize} $i\leftarrow 1$\;
\While{$i\leq d^\ast$ and $\tau\geq 0$}
{
$\mathcal{G}_i^{\ast}\leftarrow \{j,\forall j\in \mathcal{G}_i,c_j\leq \bar{v}\}$\;
\eIf{$\tau\geq \sum_{j\in\mathcal{G}_i^\ast}a_j$}
{
$\tau\leftarrow \tau - \sum_{j\in \mathcal{G}_j^\ast} a_j$\;
\For{$s\in \mathcal{G}_i^\ast$}{$\pi_s\leftarrow a_s,x_s\leftarrow \pi_s \cdot \bar{v}$}
}
{
\tcp*[h]{Add virtual node $\phi$}\;
$\phi:t_{\phi}\leftarrow(\tau,\bar{v},\emptyset)$,
$\mathcal{G}_i^\ast\leftarrow \mathcal{G}_i^\ast\cup \{\phi\}$\;
$\boldsymbol{\pi}^\ast(\mathcal{G}_i^\ast)\leftarrow\arg\min_{\boldsymbol{\pi}^\prime}\boldsymbol{\pi}(\mathcal{G}_i^\ast)\boldsymbol{c}_{\mathcal{G}_i^\ast}$\;
\For{$j\in \mathcal{G}_i^\ast$}
{
$\pi_j\leftarrow \boldsymbol{\pi}^\ast_j(\mathcal{G}_i^\ast),\tau \leftarrow \tau - \pi_j$\;
$\boldsymbol{\pi}^\ast(\mathcal{G}_i^{-j})\leftarrow\arg\min_{\boldsymbol{\pi}^\prime}\boldsymbol{\pi}(\mathcal{G}_i^{-j})\boldsymbol{c}_{\mathcal{G}_i^{-j}}$\;
$x_j\leftarrow \boldsymbol{\pi}^\ast(\mathcal{G}_i^{-j})\boldsymbol{c}_{\mathcal{G}_i^{-j}}-(\boldsymbol{\pi}^\ast(\mathcal{G}_i^\ast)\boldsymbol{c}_{\mathcal{G}_i^{\ast}}-\pi_jc_j)$\;
}

} 
} 
\caption{Networked Reverse Auction with Homogeneous Goods (RAN-HM)}
\label{pdam-hm}
\end{algorithm}}}

\subsection{Mechanism Characterization}
Next, we introduce a novel networked reverse auction mechanism for homogeneous tasks(RAN-HM). It (Alg.\ref{pdam-hm}) goes through two phases. Firstly, the networked market is divided by Alg.\ref{nmd} into sub-markets. Then $p$ conducts her reverse auction  
from the sub-markets one by one, from near to far. In each $\mathcal{G}_i$, $p$ selects all suppliers with unit cost no greater than $\bar{v}$ as $\mathcal{G}_i^\ast$ (line 4), if the total of supply of $\mathcal{G}_i^\ast$ does not fulfill the current procurement requirement, then $p$ will select all suppliers in $\mathcal{G}_i^\ast$ and pay the reserve cost for each unit good (line 5-9). Otherwise, if the total supply in the current sub-market overfill the current requirement, then $p$ selects those suppliers with the lowest unit cost by allocation in Definition \ref{definition:optimal_allocation}. Each selected supplier's payment equals the social cost increase inflicted on other suppliers by her presence (line 11-17). 

\begin{theorem}
RAN-HM is individual rational. 
\end{theorem}
\begin{proof}
Assuming it is in $\mathcal{G}_{d^\prime}$ that suppliers in $\mathcal{G}_{d^\prime}^\ast$ oversupplies the left $\tau$. For any supplier selected before $\mathcal{G}_{d^\prime}$, her utility is $u_i=a_i(\bar{v}-c_i)\geq 0$, since mechanism only chooses those suppliers whose unit cost is lower than $\bar{v}$. For any $j$ in $\mathcal{G}_{d'}^\ast$, if she is selected, her utility is $u_j=x_j - \pi_j c_j= \boldsymbol{\pi}^\ast(\mathcal{G}_i^{-j})\boldsymbol{c}_{\mathcal{G}_i^{-j}} - \boldsymbol{\pi}^\ast(\mathcal{G}_i^\ast)\boldsymbol{c}_{\mathcal{G}_i^{\ast}} \geq 0$. If $\boldsymbol{\pi}^\ast(\mathcal{G}_i^\ast)\boldsymbol{c}_{\mathcal{G}_i^{\ast}} < \boldsymbol{\pi}^\ast(\mathcal{G}_i^{-j})\boldsymbol{c}_{\mathcal{G}_i^{-j}}$, 
then it contradicts that $\boldsymbol{\pi}^\ast$ minimizes social cost. So for all suppliers selected in RAN-HM, truthfully reporting type will never obtain negative utility.
\end{proof}

\begin{theorem}
RAN-HM is incentive compatible.
\end{theorem}
\begin{proof}
Lemma \ref{lem1} has shown that allocation rule of RAN-HM is diffusion monotone. For valuation dimension, every time mechanism greedily chooses lowest cost supplier. So we can see RAN-HM satisfies allocation-monotone.
Assuming $\mathcal{G}_{d^\prime}^\ast$ oversupplies $\tau$ and for suppliers selected before $\mathcal{G}_{d^\prime}$, their critical unit winning cost is $\bar{v}$; In $\mathcal{G}_{d^\prime}$, since no supplier in $\mathcal{G}_{d^\prime}$ could change left $\tau$ when $p$ enter $\mathcal{G}_{d^\prime}$ and RAN-HM's allocation and payment in $\mathcal{G}_{d^\prime}$ apply that of VCG mechanism, Implementation of RAN-HM in $\mathcal{G}_{d^\prime}$ satisfies incentive compatible. 
In general, for all suppliers, truthfully reporting type always maximizes their own utility. 
\end{proof}

\begin{theorem}
RAN-HM is weakly budget balanced.
\end{theorem}
\begin{proof}
Since RAN-HM only select suppliers with $c_j\geq \bar{v}$ into account we directly get:
$u_p = \tau \bar{v} - \sum_{i\in N}\boldsymbol{x}_i \geq \tau\bar{v} - \sum_{i\in N} \boldsymbol{\pi}_i \bar{v} = 0$. 
\end{proof}

\section{Networked Procurement Auction for Heterogeneous Goods}

Consider that the requester procures for a set of heterogeneous goods $\boldsymbol{\tau}=\{\tau_1,\cdots,\tau_k\}$, and her private reserve costs for these goods are $\bar{\boldsymbol{v}}=\{\bar{v}(\tau_1),\cdots,\bar{v}(\tau_k)\}$. She broadcasts her purchase requirement to her own friends. Then the networked market gets expanded. The profile of agents' types is $\boldsymbol{t}=(t_1,\cdots,t_N)$. Note that $t_i=(a_i,c_i,r_i)$, where now $a_i$ is the \textit{set} of heterogeneous goods that $i$ can provide and $c_i$ is the total cost for goods she can provide.
When applying VCG mechanism, one problem is that the deficit problem still exists; while another problem is we prove that implementing VCG mechanism here is computationally intractable. 

\begin{theorem}\label{the4}
Finding social cost minimization allocation for networked combinatorial reverse auction is $\mathcal{NP}$-hard.
\end{theorem}

\begin{proof}
Firstly, we introduce the \textit{weighted set cover problem}   \cite{conf/coco/Karp72}: one finite universe set $U=\{e_1,\cdots,e_n\}$ of $n$ elements, $S=\{s_1,\cdots,s_m\}$ of $m$ subsets of $U$ with weights $W=\{w_1,\cdots,w_m\}$. Decision version is whether there exists a set cover with weights no greater than a constant $\alpha$. Assuming we have a polynomial time algorithm for each instance of the social cost minimization problem and weighted set cover problem is $\mathcal{NP}$\textit{-complete}. Given one networked combinatorial reverse auction: $G=(V,E),\boldsymbol{t}_p=(\boldsymbol{\tau},\bar{\boldsymbol{v}}),\boldsymbol{t}=(t_1,\cdots,t_N)$ and the minimum value of social cost $\beta$, then define $U=\{\tau_1,\cdots,\tau_k\}$, $S=\{a_1,\cdots,a_N,\{\tau_1\},\cdots,\{\tau_k\}\}$ with weights $W=\{c_1,\cdots,c_N,\bar{v}(\tau_1),\cdots,\bar{v}(\tau_k)\}$. It is not hard to check that $\exists$ a minimum social cost $\beta$ $\Leftrightarrow$ $\exists$ a set cover with total weights equals $\beta$, and directly comparing $\alpha$ and $\beta$ gives an answer to the decision version weighted set cover problem. All reduction conduct in polynomial time, therefore we would have polynomial time algorithm to solve the weighted set cover problem, which contradicts. So finding social cost minimization allocation for networked combinatorial reverse auction is $\mathcal{NP}$-hard.
\end{proof}

Due to the $\mathcal{NP}$\textit{-hardness}, to provide a tractable allocation, we can relax the objective of social cost minimization. We devise the first networked combinatorial reverse auction mechanism and introduce it in Alg.\ref{pdam-ht}. This algorithm also utilizes the networked market division in Alg.\ref{nmd}. 
In each independent sub-market, we repeatedly and greedily choose the supplier with the current highest non-negative marginal utility for requester $p$ (see Definition \ref{mv}). When there does not exist one non-negative marginal utility supplier in the current sub-market, $p$ moves to the next one and stops until all goods get procured or all markets are visited. 

To illustrate RAN-HT (Alg.\ref{pdam-ht}). Firstly, we give the definition of \textit{marginal valuation} and \textit{marginal utility}. 
\begin{definition}\label{mv}
Let $\mathcal{L}$ denote an arbitrary set of heterogeneous goods, and denote its corresponding reserve costs as $\bar{\boldsymbol{v}}$. Given a supplier $i$ with type $t_i=(a_i,c_i,r_i)$, $i$'s marginal valuation regarding $\mathcal{L}$ is defined as $mv_i(\mathcal{L})=\sum_{\tau\in \mathcal{L}\cap a_i}\bar{v}(\tau)$. Moreover, $i$'s marginal utility $\tilde{v}_i(\mathcal{L})$ equals her marginal valuation minus her own product cost, i.e. $\tilde{v}_i(\mathcal{L}) = mv_i(\mathcal{L}) - c_i$.
\end{definition}

In Alg.\ref{pdam-ht}, RAN-HT firstly calls Alg.\ref{nmd} to divide the whole market $\mathcal{G}$ into sub-markets (line 1). The first loop (line 3) means, $p$ procures goods sequentially from sub-market $\mathcal{G}_1$ to sub-market $\mathcal{G}_{d^\ast}$. For each loop $i$, requester $p$ enters sub-market $\mathcal{G}_i$, RAN-HT immediately finds all uncovered goods $\mathcal{L}_i$ (line 4) and chooses suppliers with current maximum non-negative marginal utility, and update $\mathcal{L}_i$ (line 5-8). Denote $\mathcal{G}_i^\ast$ as those selected suppliers from $\mathcal{G}_i$. To decide $x_j$, RAN-HT finds winners $\mathcal{W}$ within the market without $j$, which is denoted as $\mathcal{G}_i^{-j}$ (line 12-16) and set $x_j$ as the maximal reported cost that can beat all winners in $\mathcal{W}$ (line 18-23). If there exists no extra winner under $\mathcal{G}_i^{-j}$, then $p$ would pay $j$ her marginal valuation regarding $\mathcal{L}_i$ (line 17). Algorithm terminates when all items are covered or all sub-markets are visited.

{\small{
\begin{algorithm}[h]
\SetAlgoLined
\SetKwInOut{Input}{Input}\SetKwInOut{Output}{Output}
\SetKwFunction{NetworkedMarketDivision}{NetworkedMarketDivision}
\Input{$G=(V,E),\boldsymbol{t}_p=(\boldsymbol{\tau},\bar{\boldsymbol{v}}),\boldsymbol{t}=(t_1,\cdots,t_N)$}
\Output{$\boldsymbol{\pi},\boldsymbol{x}$}
$(\mathcal{G}_1,\cdots,\mathcal{G}_{d^\ast})\leftarrow$ \NetworkedMarketDivision{$G$}\;
Initialize $i\leftarrow 1,\mathcal{S}\leftarrow \emptyset$\;
\While{$i\leq d^\ast$}
{
$\mathcal{L}_i\leftarrow \boldsymbol{\tau}-\cup_{q\in\mathcal{S}}a_q,\mathcal{G}_i^\ast\leftarrow \emptyset$,$\mathcal{L}_i^\prime\leftarrow \mathcal{L}_i$\;
\tcp*[h]{Allocation Phase for $\mathcal{G}_i$}\;
\While{$\mathcal{G}_i^\ast\neq \mathcal{G}_i$ and $\tilde{v}_j(\mathcal{L}_i^\prime)\geq 0$}
{
$j\leftarrow \arg\max_{\mathcal{G}_i\setminus\mathcal{G}_i^\ast}\tilde{v}_j$\;
$\mathcal{G}_i^\ast\leftarrow \mathcal{G}_i^\ast \cup \{j\}$,$\mathcal{L}_i^\prime \leftarrow \mathcal{L}_i^\prime \setminus \cup_{s\in \mathcal{G}_i^\ast}a_s$\;
}
$\mathcal{S}\leftarrow \mathcal{S} \cup \mathcal{G}_i^\ast$\tcp*[h]{Update Selected Set}\;
\tcp*[h]{Payment Phase for $\mathcal{G}_i^\ast$}\;
\For{$j\in \mathcal{G}_i^\ast$}
{
$\pi_j\leftarrow 1,x_j\leftarrow 0$\;
$\mathcal{W}\leftarrow\emptyset,\mathcal{G}_i^{-j}\leftarrow \mathcal{G}_i\setminus \{j\}, \mathcal{L}_i^j\leftarrow \mathcal{L}_i$\;
\While{$\mathcal{W}\neq \mathcal{G}_i^{-j}$ and $\tilde{v}_{\ell}(\mathcal{L}_i^j)\geq 0$}
{
$\ell \leftarrow \arg\max_{\mathcal{G}_i^{-j}\setminus \mathcal{W}} \tilde{v}_\ell(\mathcal{L}_i^j)$\;
$x_j\leftarrow \max(x_j,mv_j(\mathcal{L}_i^j)-\tilde{v}_{\ell}(\mathcal{L}_i^j)$\;
$\mathcal{W}\leftarrow \mathcal{W}\cup \{\ell\},\mathcal{L}_i^j\leftarrow \mathcal{L}_i^j\setminus \cup_{s\in \mathcal{W}}a_s$\;
}
\If{$\mathcal{W}=\emptyset$}{$x_j\leftarrow mv_j(\mathcal{L}_i^j)$}
\If{$\mathcal{W}=\mathcal{G}_i^{-j}$}{$x_j\leftarrow\max(x_j,mv_j(\mathcal{L}_i^j))$}

}
}
\caption{Networked Reverse Auction with Heterogeneous Goods (RAN-HT)}
\label{pdam-ht}
\end{algorithm}}}

\begin{example}\label{example-2}
In Fig.\ref{fig2}, requester $p$ is with $\boldsymbol{t}_p=(\boldsymbol{\tau},\bar{\boldsymbol{v}})$, where $\boldsymbol{\tau}=\{a,b,c,d,e,f,g,h,j,k\}$, $\bar{\boldsymbol{v}}=\{a:3,b:5,c:4,d:8,e:10,f:7,g:12,h:9,j:6,k:5\}$. Firstly $p$ releases $\boldsymbol{\tau}$ to her own friends: $s_1,s_2,s_3$, then $s_1,s_2,s_3$ also invite their own friends to join this market. Then friends diffuse to friends, which finally forms a large networked market $G$. 
\end{example}
To run the procedure of RAN-HT in Example \ref{example-2}, firstly, we divide the whole market $G$ into three sub-markets: $\mathcal{G}_1=\{s_1,s_2,s_3\},\mathcal{G}_2=\{s_4,s_5,s_6,s_7\}$ and $\mathcal{G}_3=\{s_8\}$. Then the requester makes procurement along the three sub-markets. In $\mathcal{G}_1$,
we can calculate three suppliers' marginal utility for $p$: $\tilde{v}_{s_1}(\mathcal{L}_1)=26,\tilde{v}_{s_2}(\mathcal{L}_1)=17,\tilde{v}_{s_3}(\mathcal{L}_1)=19$, thus we add $s_1$ into $\mathcal{G}_1^\ast$, then those unselected suppliers marginal utility change to $\tilde{v}_{s_2}(\mathcal{L}_1\setminus a_{s_1})=7,\tilde{v}_{s_3}(\mathcal{L}_1\setminus a_{s_1})=4$. So next we add $s_2$ into $\mathcal{G}_1^\ast$ and then we can find that $s_3$'s marginal utility to $p$ becomes negative, so we add all $s_1,s_2$ in $\mathcal{G}_1^\ast$ into $\mathcal{S}$. for $s_1$, $\{s_2,s_3\}$ will be winners in $\mathcal{G}_1^{-s_1}$, $x_{s_1}=\max\{38-19,23-3,23\}=23$; for $s_2$, $\{s_1,s_3\}$ will be selected in $\mathcal{G}_1^{-s_2}$, $x_{s_2}=\max\{26-26,16-4,12\}=12$. Then the requester leaves $\mathcal{G}_1$ and enter $\mathcal{G}_2$. Note at this time, $\mathcal{L}_2=\{a,g\}$, for suppliers in sub-market $\mathcal{G}_2^\ast$, $\tilde{v}_{s_4}(\mathcal{L}_2)=-8,\tilde{v}_{s_5}(\mathcal{L}_2)=-6,\tilde{v}_{s_6}(\mathcal{L}_2)=5,\tilde{v}_{s_7}(\mathcal{L}_2)=-2.$ So in $\mathcal{G}_2$, $\mathcal{G}_2^\ast=\{s_6\}$, then add $s_6$ into $\mathcal{S}$. For $s_6$, nobody will be winner in $\mathcal{G}_2^{-s_6}$, so $\mathcal{W}=\emptyset$ and $x_{s_6}=mv_{s_6}(\mathcal{L}_2)=12$. Finally $p$ enter $\mathcal{G}_3$ and $\mathcal{L}_3=\{a\}$, for the unique supplier $s_8$ in $\mathcal{G}_3$, $\tilde{v}_{s_8}(\mathcal{L}_3)=-8$. So the algorithm terminates and the final selected suppliers are $\mathcal{S}=\{s_1,s_2,s_6\}$, i.e. $\boldsymbol{\pi}_{s_1}=1,\boldsymbol{\pi}_{s_2}=1,\boldsymbol{\pi}_{s_6}=1$ and their payments are $\boldsymbol{x}=\{x_{s_1}=23,x_{s_2}=12,x_{s_6}=12\}$. From the perspective of $p$, all of procurement tasks except $a$ are covered by the selected three suppliers $\{s_1,s_2,s_6\}$. Procuring under such a networked supply market, both social cost ($69\to 31$) and the requester's cost ($69\to 50$) get decreased. 

\begin{figure}[h]
	    \centering
    	    \includegraphics[scale=1.0]{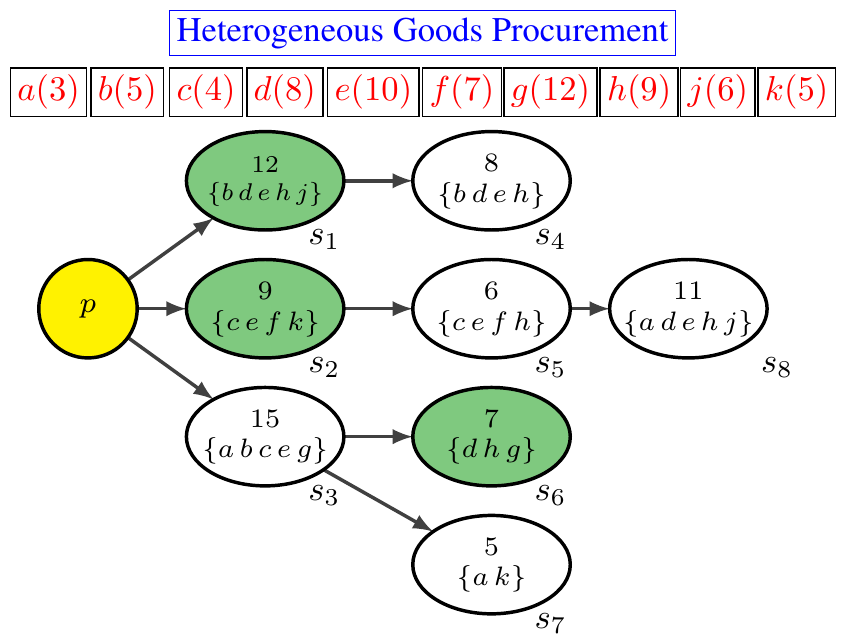}
    \caption{Running Example for RAN-HT, yellow circle is the requester $p$ and each oval represents one supplier $s_i$, the number in the upper part of each oval is the supplier $s_i$'s total cost for providing the bundle of goods. The set in the lower part of each oval is the bundle of goods that supplier $s_i$ could provide. Each directed arrow indicates the direction of information diffusion. Letters and numbers in the black boxes indicates the reverse cost of each good.}
     \label{fig2}
\end{figure}

\begin{theorem}
Networked Reverse Auction Mechanism for Heterogeneous Goods (RAN-HT) is individual rational.
\end{theorem}

\begin{proof}
For not selected suppliers, truthful report leads to $u_i=0$.
For $j\in \mathcal{G}_i^\ast$, 
1. $\mathcal{W}=\emptyset$ under $\mathcal{G}_i^{-j}$, then $u_j=x_j-c_j=mv_j(\mathcal{L}_i)-c_j\geq 0$;
2. $x_j=mv_j(\mathcal{L}_i^{j})-\tilde{v}_\ell(\mathcal{L}_i^j)$ represents $\ell$ would replace $j$'s wining position in $\mathcal{G}_i^{-j}$. Suppliers selected before $j$ will still be preferred to be selected, $j$ still could beat $\ell$ under $\mathcal{L}_i^j$: $mv_j(\mathcal{L}_i^j)-c_j\geq mv_\ell(\mathcal{L}_i^j)-c_\ell$, i.e. $c_j\leq mv_j(\mathcal{L}_i^j)-\tilde{v}_\ell(\mathcal{L}_i^j)=x_j$.
3. $x_j=mv_j(\mathcal{L}_i^j)$, then we know $j$ must be the last selected supplier in $\mathcal{G}_i$. $u_j=mv_j(\mathcal{L}_i^j)-c_j=mv_{j}(\mathcal{L}_i)-c_j\geq 0$.
Since $\mathcal{S}=\cup_{i=1}^{d^\ast}\mathcal{G}_{i}^\ast$, for all selected suppliers, reporting $(a_j,c_j,r_i)$ always guarantees non-negative utility. 
\end{proof}

\begin{theorem}
Networked Reverse Auction Mechanism for Heterogeneous Goods (RAN-HT) is incentive compatible.
\end{theorem}

\begin{proof}
When fixed $c_j$, due to the division of the whole market, $\forall\,r_j^\prime \subseteq r_j, \boldsymbol{\pi}_j(\emptyset) = \boldsymbol{\pi}_j(r_j^\prime) = \boldsymbol{\pi}_j(r_j)$. $\forall\, c_j^\prime < c_j$, since RAN-HT always chooses the current highest marginal utility supplier, $\tilde{v}_j(c_j,\mathcal{L}_i^\prime)=mv_j(\mathcal{L}_i^\prime)-c_j \leq mv_j(\mathcal{L}_i^\prime)-c_j^\prime \leq mv_j(\mathcal{L}_i'')-c_j^\prime = \tilde{v}_j^\prime(c_j^\prime,\mathcal{L}_i'')$, thus  $\boldsymbol{\pi}_j(c_j^\prime)\geq \boldsymbol{\pi}_j(c_j)$.
Then $\forall\, c_j^\prime<c_j,r_j^\prime\subseteq r_j$, we can know that $\boldsymbol{\pi}_j(c_j^\prime,r_j^\prime)\geq \boldsymbol{\pi}_j(c_j,r_j)$ which satisfies the monotonicity of allocation in diffusion auctions, so next we only need to prove that for each selected supplier $j$, her payment always equals her critical winning cost. Due to the market division, her critical winning cost depends only on her reporting cost $c_j$. So we only discuss about $c_j$ here, $\forall\, j\in \mathcal{S}$, according to the allocation monotonicity, low reporting $c_j$, $j$ will still be the winner and her payment is independent with her own type. 
Next consider $c_j^\prime >x_j$: (1). $x_j=mv_j(\mathcal{L}_i)<c_j^\prime$ means $\tilde{v}_j(\mathcal{L}_i)$, $j$ will lose; (2). $x_j=mv_j(\mathcal{L}_i^j)-\tilde{v}_\ell(\mathcal{L}_i^j)$, high report $c_j^\prime > x_j$ will make $j$ lose to $\ell$ under $\mathcal{L}_i^j$, on the other hand, $x_j = \max_{\ell^\prime}\{mv_j(\mathcal{L}_i^j)-\tilde{v}_{\ell^\prime}(\mathcal{L}_i^j)\}$, $j$ also will lose to all winners after $\ell$, so reporting $c_j^\prime$, $j$ will never be selected; (3) if $x_j=mv_j(\mathcal{L}_i^j)$, then $j$ is the last selected supplier, obviously, reporting $c_j^\prime > x_j$ would lose. Combining the monotone allocation rule with critical winning cost as payment, we can deduce the IC of RAN-HT.
\end{proof}

\begin{theorem}
Networked Reverse Auction Mechanism for Heterogeneous goods (RAN-HT) is weakly budget balanced.
\end{theorem}

\begin{proof}
Generally, we can divide all procurement tasks into two parts: $u_p=u_p^\mathcal{S}+u_p^{\bar{\boldsymbol{v}}}$. The former is the utility gains from allocating tasks to suppliers while the latter is finished by $p$ herself. Obviously $u_p^{\bar{\boldsymbol{v}}}=0$, so we only need to consider $u_p^\mathcal{S}$.
{\small{\[
\begin{split}
    u_p^{\mathcal{S}} &= \sum_{i=1}^{d^\ast} u_p^{\mathcal{G}_i}= \sum_{i=1}^{d^\ast} \left( \sum_{\tau\in \mathcal{L}_i}\bar{\boldsymbol{v}}(\tau)-\sum_{j\in \mathcal{G}_i^\ast}x_j \right)\\
    &=\sum_{i=1}^{d^\ast}\left( \sum_{\tau\in \mathcal{L}_i}\bar{\boldsymbol{v}}(\tau) - \sum_{j\in \mathcal{G}_i^\ast}\left(mv_j(\mathcal{L}_i^j)-\tilde{v}_{\ell}(\mathcal{L}_i^j)\right)\right)\\
    &\geq \sum_{i=1}^{d^\ast}\left( \sum_{\tau\in \mathcal{L}_i}\bar{\boldsymbol{v}}(\tau) - \sum_{j\in \mathcal{G}_i^\ast} mv_j(\mathcal{L}_i^j)\right)\\
    &= \sum_{i=1}^{d^\ast}\left( \sum_{\tau\in \mathcal{L}_i}\bar{\boldsymbol{v}}(\tau) - \sum_{j\in \mathcal{G}_i^\ast} \sum_{\tau^j\in\mathcal{L}_i^j\cap a_j}\bar{\boldsymbol{v}}(\tau^j)\right)=0\\
\end{split}
\]}}

Intuitively, requester $p$ only chooses those suppliers with non-negative marginal utilities and for each selected supplier, her payment could not go beyond her marginal valuation on the remaining tasks. So $p$ never suffers from deficit.
\end{proof}

\section{Simulations}
To evaluate the perform of our mechanisms RAN-HM and RAN-HT, we have implemented the procurement markets for both homogeneous and heterogeneous tasks. For procuring homogeneous tasks, we conduct experiments for three different mechanisms: ND-VCG (VCG mechanism without diffusion, i.e. market only contains the requester and her neighbourhood suppliers.), D-VCG (VCG mechanism over the whole market), RAN-HM (the novel mechanism introduced in Alg.\ref{pdam-hm}). For procuring heterogeneous tasks, due to the computational intractability, we compare two different mechanisms, one greedy algorithm(Implementing RAN-HT in the local market) and our RAN-HT. Also, to evaluate the robustness of our mechanisms over different markets, we compared different scenarios with different scale of markets from three aspects: different social networks, different scale of tasks and different scale of suppliers. To get more reliable results, $20$ times experiments are conducted under each minimum experimental unit.

\subsection{Experimental Results of Procurement Homogeneous Tasks}
\subsubsection{Social Networks Domain}
For different social networks, we generate different social networks with a random graph algorithm (See supplementary for details) where we set one parameter $prob$ to decide whether to generate one edge between two random chosen vertices or not. Procedures of the algorithm are shown in Alg.\ref{grg}:

{\small{
\begin{algorithm}
\SetAlgoLined
\SetKwInOut{Input}{Input}\SetKwInOut{Output}{Output}
\SetKwFunction{Random}{Random}
\SetKwFunction{IsConnected}{IsConnected}
\Input{Suppliers Amount: $N$, Probability of generating edge: $prob$}
\Output{Social Network $G=(N+1,E)$}
Initialize \textbf{flag}$\leftarrow 0,G=(N+1,E = \emptyset),i\leftarrow 1,j \leftarrow 1$\;
\While{\textbf{flag} $\neq 1$}
{
\For{$i < N+1$}
{
\For{$i < j < N+1$}
{
 $p \leftarrow $ \Random(0,1)\;
 \If{$p < prob$}
 {
 $E \leftarrow E \cup (i,j)$
 }
}
}
\If{\IsConnected}
{
\textbf{flag} $= 1$\;
}
}
\caption{Generating Random Graph}
\label{grg}
\end{algorithm}}}

For different scale of tasks, we fix the probability $prob$ for generating networks and the number of suppliers in our market; similarly, we fix the probability and the number of the tasks to explore the influence of different supplier amounts. Detailed parameter settings are shown as follows: We range the probability value $prob$ varies from $[0.05,0.30]$. Obviously, the density of the graph increases along with the increase of the probability for generating edges. Besides, we fix three different market sizes:
(1). $N=20,|\boldsymbol{\tau}|=100$; (2). $N=40, |\boldsymbol{\tau}|=200$; (3). $N=100, |\boldsymbol{\tau}|=500$. 
For each agent $i$, we make assumption that firstly, her ability $a_i$ for finishing tasks follows the uniform distribution: $a_i\sim U[1,10]$, secondly, her cost $c_i$ for finishing all tasks also follows the uniform distribution: $c_i\sim U[1,10]$.

The results of payment and social cost are shown in Figure.\ref{fig-payment-networks}:
\begin{figure}[!htbp]
    \centering
    \includegraphics[width=1\textwidth]{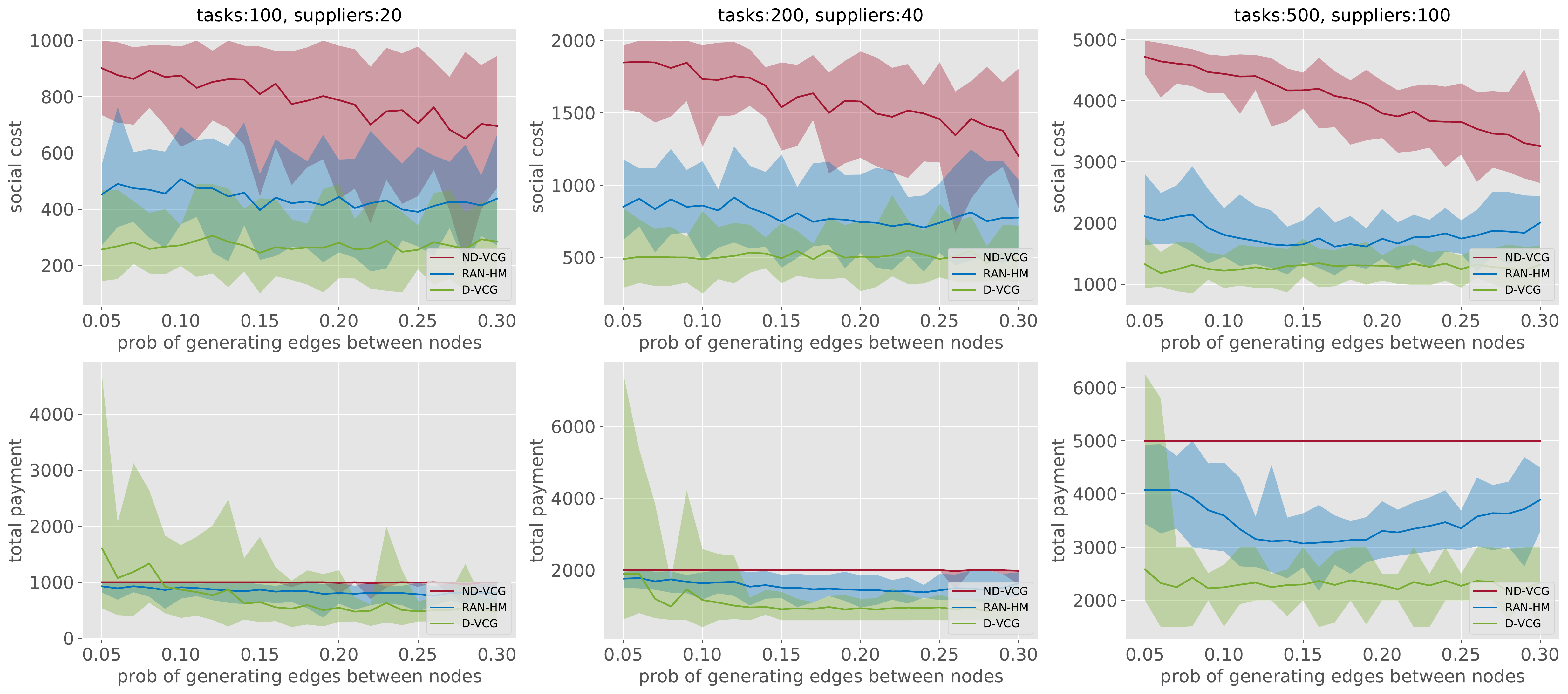}
    \caption{Social cost and payment of the requester over different social networks}
    \label{fig-payment-networks}
\end{figure}
Row one shows the relationship between social cost and the density of the social network while row two gives the requester's payment. Lines in each sub-graph represent the average value and the shaded area indicates upper and lower boundaries that the corresponding mechanism achieves. All sub-figures share the $x$-axis with different probability for connections in range $[0.05,0.30]$ with step $0.02$. 

From the perspective of social cost, a denser social network could bring more 1-hop or 2-hops suppliers around the requester, leading to lower social cost, this is most evident in ND-VCG. For RAN-HM, although it proceeds along the leveled markets one by one, it achieves quite good social cost in practice. For D-VCG, it is costly to gain the optimal social cost, it is obvious that the requester's payment via D-VCG extremely volatile according to different social networks especially when the market is not large enough (i.e. pursuing optimal social cost is extremely costly under small market). This instability weakens as the market expands and the density of social networks increases. Generally speaking, D-VCG dominants the ND-VCG and RAN-HM in the average sense, however, the terrible deficit elicited by D-VCG is unacceptable for the requester. Moreover, RAM-HM always outperforms ND-VCG and it catches great trade-offs between the allocation efficiency and the requester's payment, showing better robustness in any size of market. Thus, practically speaking, the RAN-HM is more promising for applications in reality. 
Besides, we conduct experiments (Figure.\ref{complete-g-tree}) for the complete graph and tree-structure graph to show that how those mechanisms work in extreme cases. 
\begin{figure}[!htbp]
    \centering
    \includegraphics[width=1\linewidth]{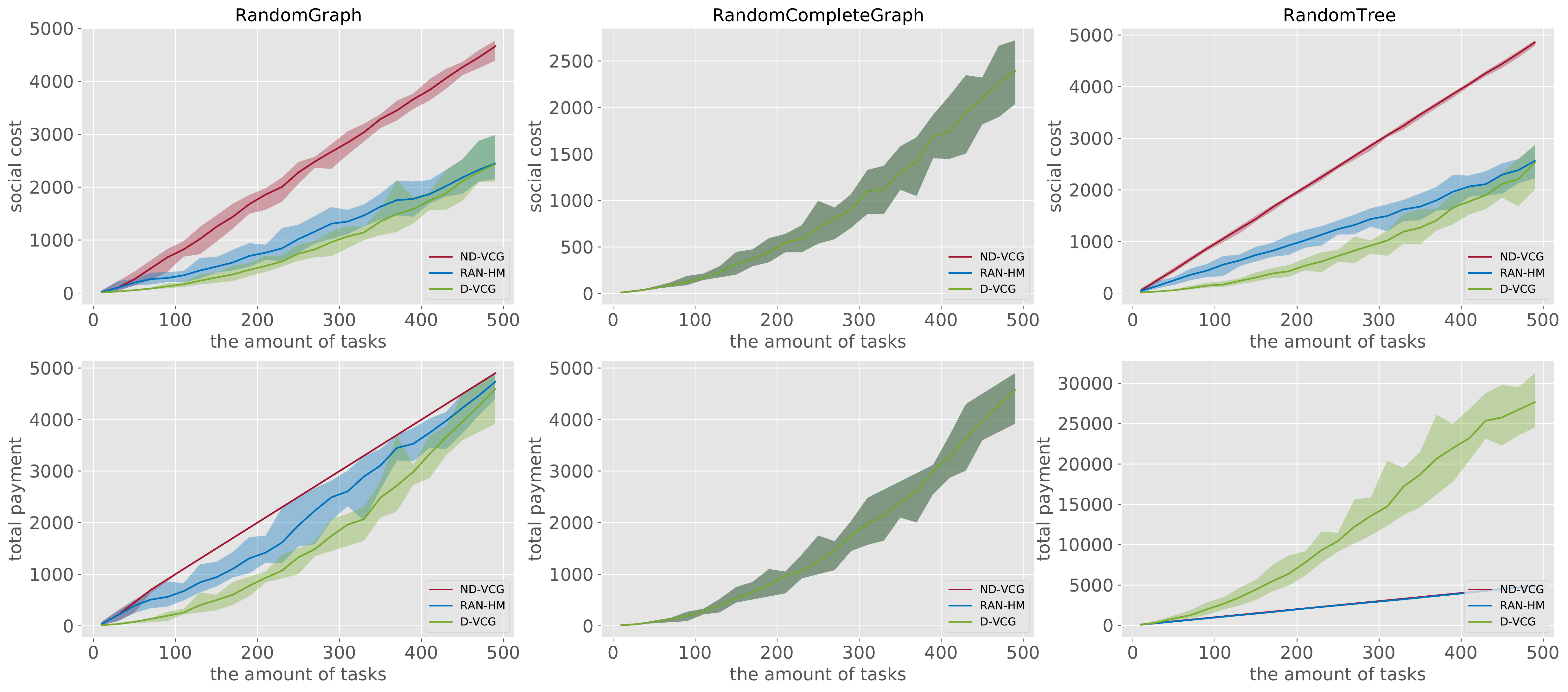}
    \caption{Social cost and total payment under random graph or complete graph and tree-structured graph}
    \label{complete-g-tree}
\end{figure}
Here we give three different types of social networks: complete graph, random graph and random tree. From the point of view of the density of the social network, the density of these three graphs are $0.2,1,0.1$ respectively. For complete graph, it is not hard to see that all suppliers are directly connected with the requester $r_p$, so the networked market degenerates to a local market, where the allocation and payment elicited by ND-VCG, D-VCG and RAN-HM take the same form (i.e. there exists no difference between these three mechanisms under any complete graphs). For random tree, the disadvantage of D-VCG will be particularly pronounced whatever the amount of tasks and suppliers are in the market.

\subsubsection{Scale of Tasks Domain}
For different task scales, we set the amount of tasks $\tau$'s range as $[10,500]$. Also, to make comparison, we firstly fix the parameter $prob$ as $0.05$ while set the amount of suppliers $N$ as $20,40,100$ respectively. 
The results of payment and social cost are shown in Figure.\ref{fig-payment-market}.
\begin{figure}[!htbp]
    \centering
    \includegraphics[width=1\textwidth]{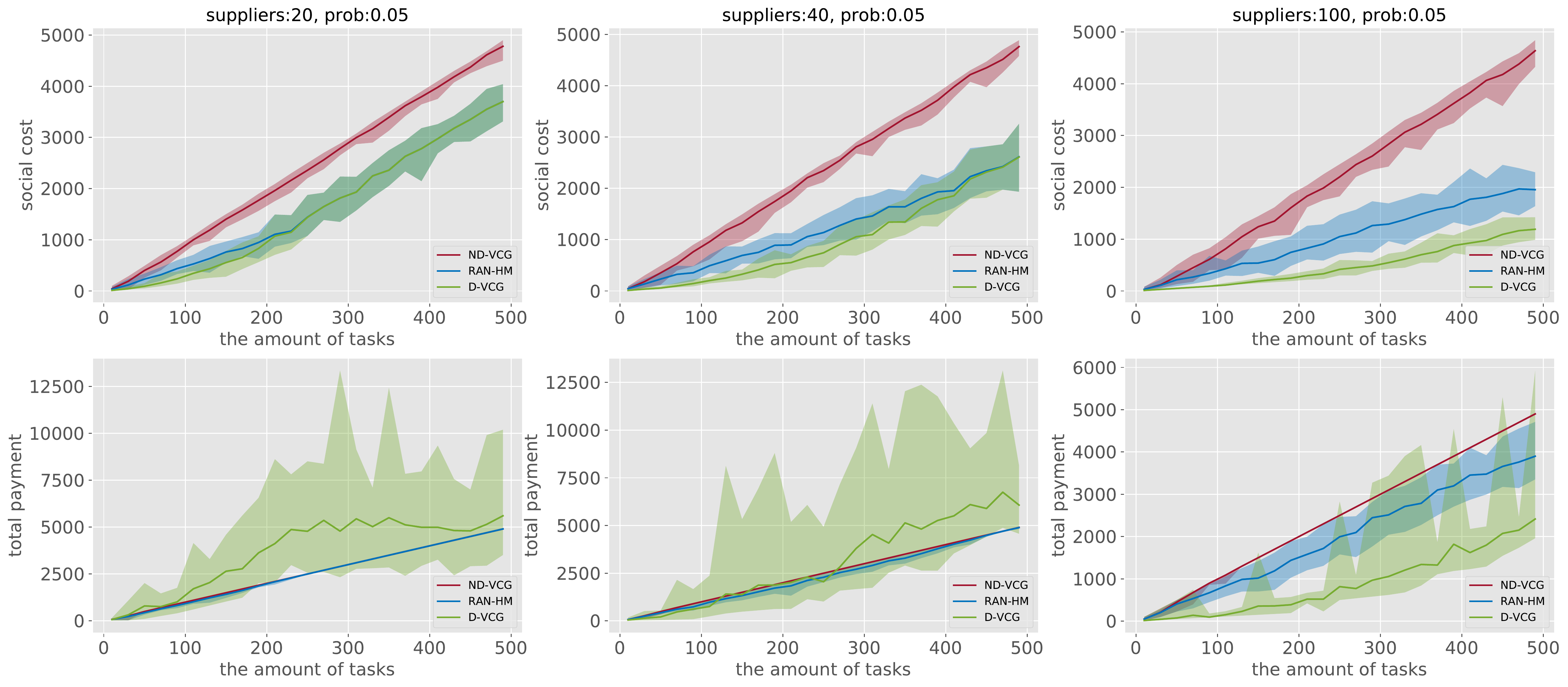}
    \caption{Social cost and payment of the requester with different amount of tasks}
    \label{fig-payment-market}
\end{figure}

Firstly, it is not hard to see that there exists an essentially linear relationship between the number of tasks and social costs. However, different mechanism has different linear ratios, for D-VCG, the linear growth rate of its social cost decreases with increasing suppliers base size. In contrast, D-VCG does not work well in the requester's payment dimension especially when the market lacks sufficient suppliers, in order to hire those suppliers with least social cost, D-VCG mechanism has to pay a considerable cost to incentive some suppliers to inviting potential suppliers with least social cost. This deficit effect diminishes as the number of suppliers in the market increases. However, as shown in the last sub-figure in Figure.\ref{fig-payment-market}, it is still risky for the deficit under some specific social networks.  For our RAN-HM, both for social cost and the total payment, RAN-HM always dominants ND-VCG, besides, it performs quite well in terms of social cost and the requester's payment in both small and large markets compared with D-VCG (i.e. near-optimal social cost, stable payment that never goes beyond budget). 


\subsubsection{Scale of Suppliers Domain}
To explore the effects of the amount of suppliers in the market, we fix the social network generation with the same probability with $prob = 0.05$, the total demand of goods or task over the whole market $|\boldsymbol{\tau}| = 500$, each supplier's ability for finishing tasks $a_i\sim U[1,10]$. For the amount of suppliers $N$, we take it in range $[10,500]$. The results of payment and social cost are shown in Figure.\ref{fig-payment-agent}. 
\begin{figure}[!htbp]
    \centering
    \includegraphics[width=1\textwidth]{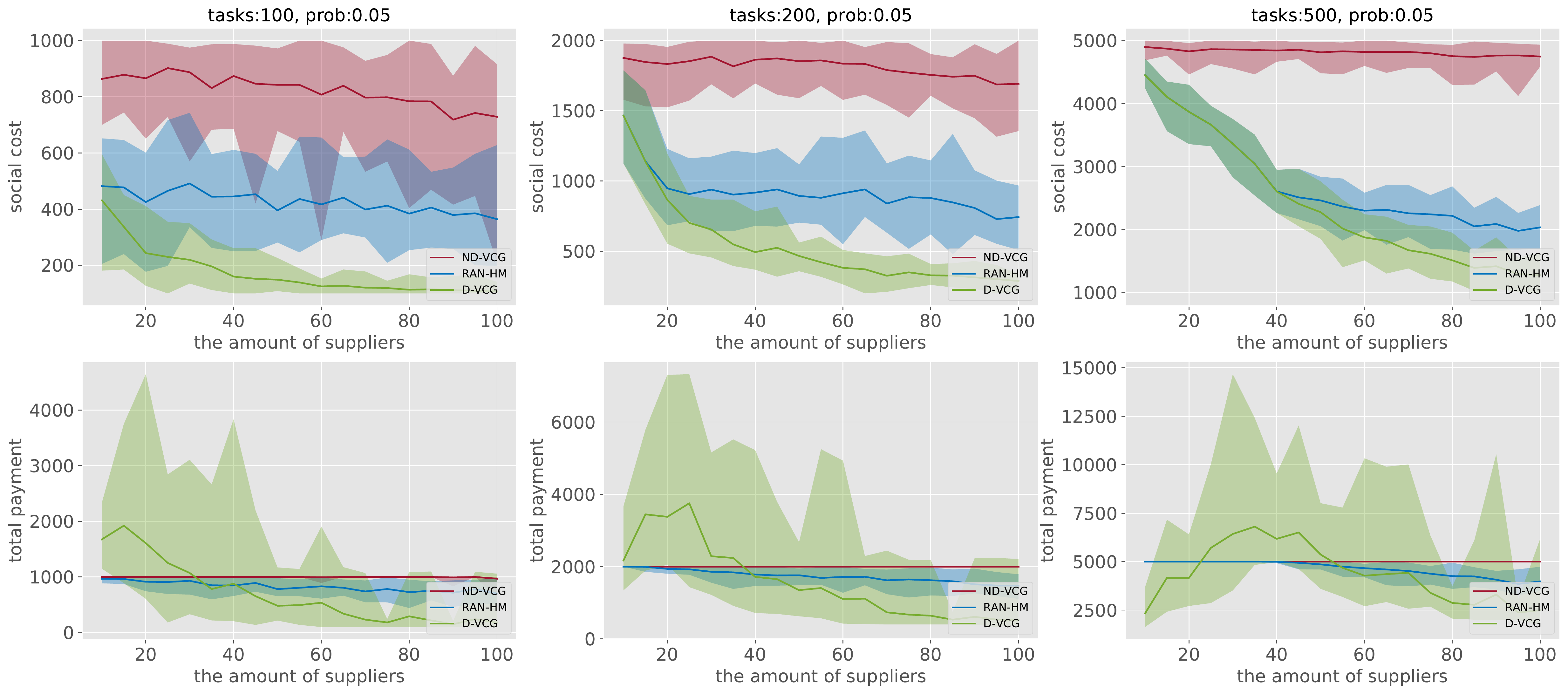}
    \caption{Social cost and payment of the requester with different amount of suppliers}
    \label{fig-payment-agent}
\end{figure}

When considering the relation between the amount of suppliers and the social cost and payment with fixed sparsity of random networks and fixed tasks. Generally speaking, with the number of suppliers increase, more potential suppliers with lower cost appear, thus both the social cost and the payment get decreased for D-VCG and the RAN-HM since both of the two mechanisms make allocation and payment over social networks. With the expansion of market size, the gap of social cost between the D-VCG and RAN-HM get larger, mainly because that RAN-HM takes the greedy form, where it always firstly satisfies suppliers who are in the former markets. From the perspective of the requester's payment, RAN-HM still shows great stability under all the potential markets: both few-tasks-many-suppliers market and many-tasks-few-suppliers market. For D-VCG, when the suppliers are not enough, it results in extreme uncertain for the payment. With the increase of supplier numbers, the instability get eased and its performance goes beyond RAN-HM in average aspect. 


\subsection{Experimental Results of Procurement Heterogeneous Tasks}
In Theorem.\ref{the4}, we have mentioned that finding social cost minimization allocation in our scenario is $\mathcal{NP}-$hard, thus it is computationally intractable with VCG mechanism under this procurement with heterogeneous goods. So here we make comparison for two different mechanisms, one is our RAN-HT while the other is the RAN-HT without diffusion settings (Note it as Greedy Algorithm). Similarly, we implement experiments for different social networks, markets and amount of suppliers. From homogeneous scenario to heterogeneous scenario, the adjusted parameters are as follows: for the requester $t_p$, we randomize each heterogeneous task's cost $\bar{v}(\tau_i)\sim U[1,10]$. for each supplier $i$, we randomize her ability $|a_i|\sim U[2,10]$ and randomly choose $|a_i|$ different tasks from $\boldsymbol{\tau}$, i.e. we assume each supplier could finish $2$ to $10$ tasks and randomly sample the specific tasks. Also, for her cost $c_i$, we assume it also follows the uniform distribution $c_i\sim U[5,20]$. To compare mechanisms under social networks with different sparsity, we still take the same generation form in Section 6.1. Besides, we fix three different market sizes:
(1). $N=20,|\boldsymbol{\tau}|=100$; (2). $N=40, |\boldsymbol{\tau}|=200$; (3). $N=100, |\boldsymbol{\tau}|=500$. 

\subsubsection{Social Networks Domain}
To make difference between different social networks, we still take the random graph algorithm in the appendix and set parameter $prob$ in range $[0.05, 0.30]$ by step $0.02$. Figure.\ref{ran-1} has shown the simulation results under procurement diffusion auctions with heterogeneous tasks or tasks.
\begin{figure}[!htbp]
    \centering
    \includegraphics[width=1\textwidth]{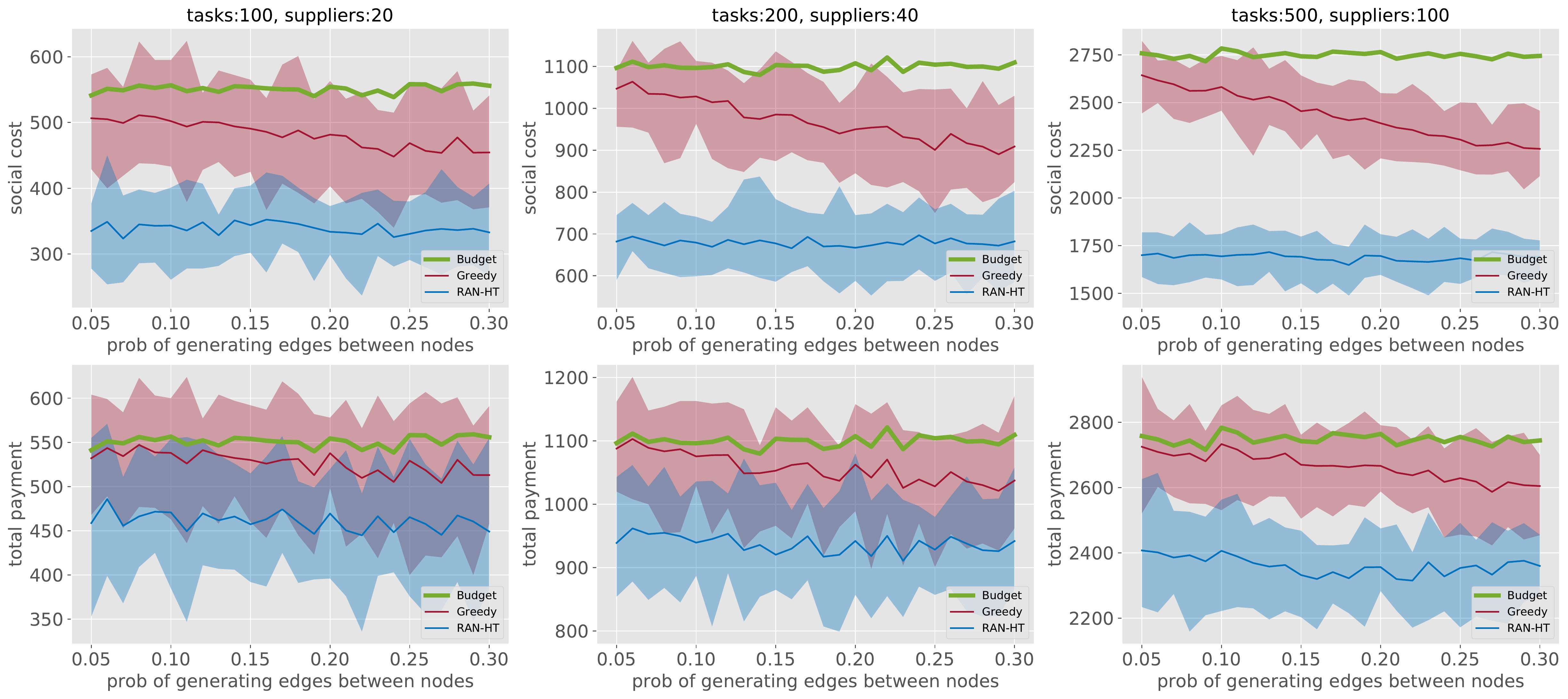}
    \caption{Social Cost and Payment under different graphs with different sparsity}
    \label{ran-1}
\end{figure}
 The bolded green line represents the total budget of the requester, i.e. the total cost for the requester to finish the tasks by herself. The red line and region represents the ordinary greedy mechanism under local markets while the blue line and region represents the RAN-HT mechanism. In row one, from the perspective of social cost, considering the greedy mechanism, with the connections between nodes in the graph become dense, the probability of the requester connecting to more neighbors become higher, leading to the decrease of social cost. In contract, RAN-HT is not sensitive with sparsity of social networks since in RAN-HT, when the current level market has no positive marginal value suppliers, it directly goes into the next market which means it tends to find some pretty good suppliers in deeper markets, which has little relation with the number of 'good' suppliers in first-level market. Also, from the perspective of market size, the performance distance between RAN-HT and the greedy mechanism gets larger with the scale of market get larger both from social cost and the requester's payment. 

\subsubsection{Scale of Tasks Domain}
To evaluate the influence of the amount of tasks that the requester posts, we implement simulations for three different markets: (1).suppliers:20; (2).suppliers:40; (3).suppliers:100. Each network is constructed by random graph with probability 0.05 to link two nodes. Also, for each market, we range the amount of tasks $|\boldsymbol{\tau}|$ with $[10,500]$. Settings not mentioned here are taken the same form as the simulation in the last part
\begin{figure}[!htbp]
    \centering
    \includegraphics[width=1\textwidth]{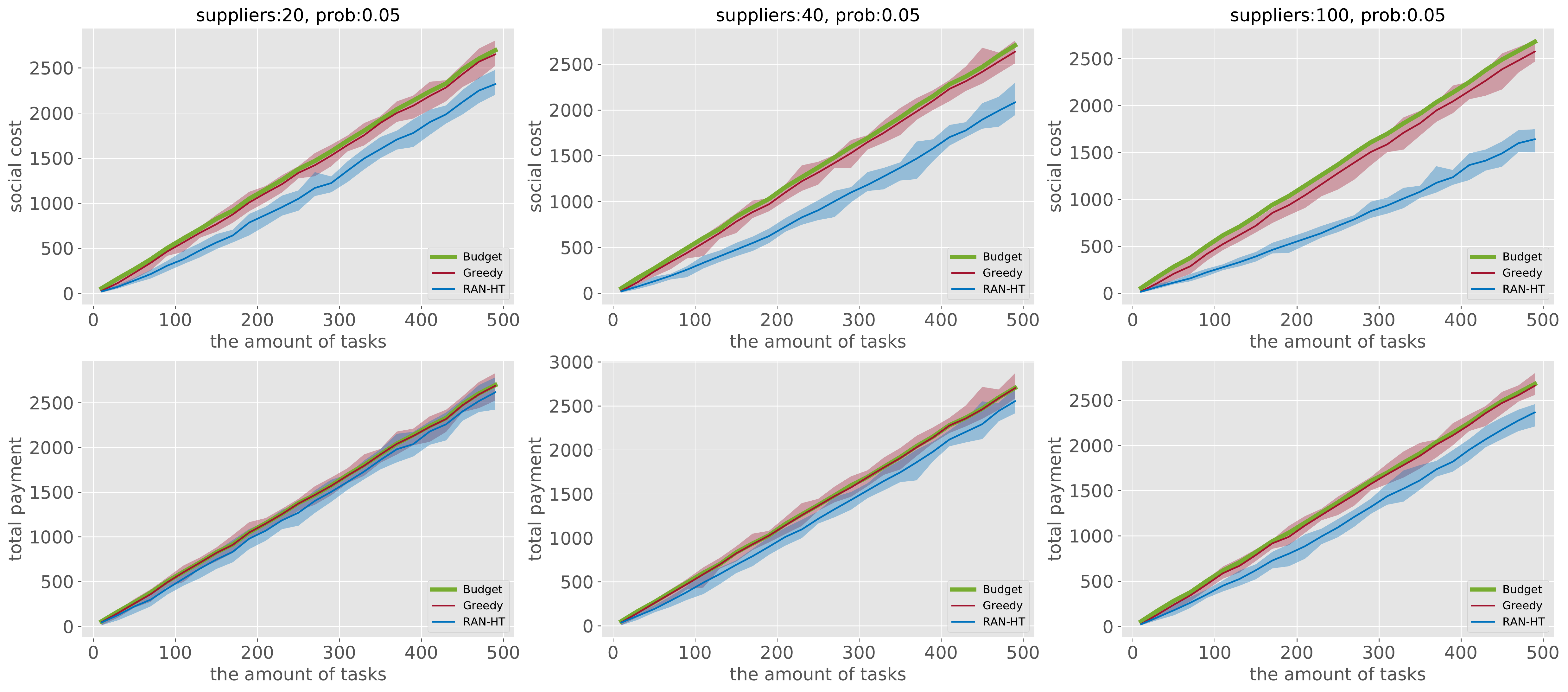}
    \caption{Social Cost and Payment under different amount of tasks}
    \label{ran-2}
\end{figure}

The relationship between social cost or total payment and the amount of total tasks are shown in Figure.\ref{ran-2}. Generally speaking, the relationship between the amount of tasks and the social cost or the total payment is essentially linear. However, different mechanisms present different linear ratio of increase. RAN-HT dominates the greedy mechanism in social cost domain. More suppliers bring more clearly dominant performance. For the payment domain, when there are not enough suppliers, the change trend of total payment of the greedy mechanism and our RAN-HT are very similar which are both around the budget line. When the amount of suppliers gets larger, RAN-HT gradually outperforms the greedy mechanism, helping the requester save some payment for recruiting. 

\subsubsection{Scale of Suppliers Domain}

The last part for our simulation is to explore suppliers' effects on the social cost and the requester's payment. Similarly, to make comparison, we conduct three different markets: (1). tasks:100; (2). tasks:200; (3). tasks: 500. Each market is constructed by random graph with probability 0.05 of linking two nodes. For suppliers domain, we range it from $10$ to $100$ with stepping by $5$. Each group of parameters simulates for $20$ times. The bolded line represents average value for each mechanism and the upper bound and lower bound of the shaded part shows the maximum value and minimum value for each test. 
\begin{figure}[!htbp]
    \centering
    \includegraphics[width=1\textwidth]{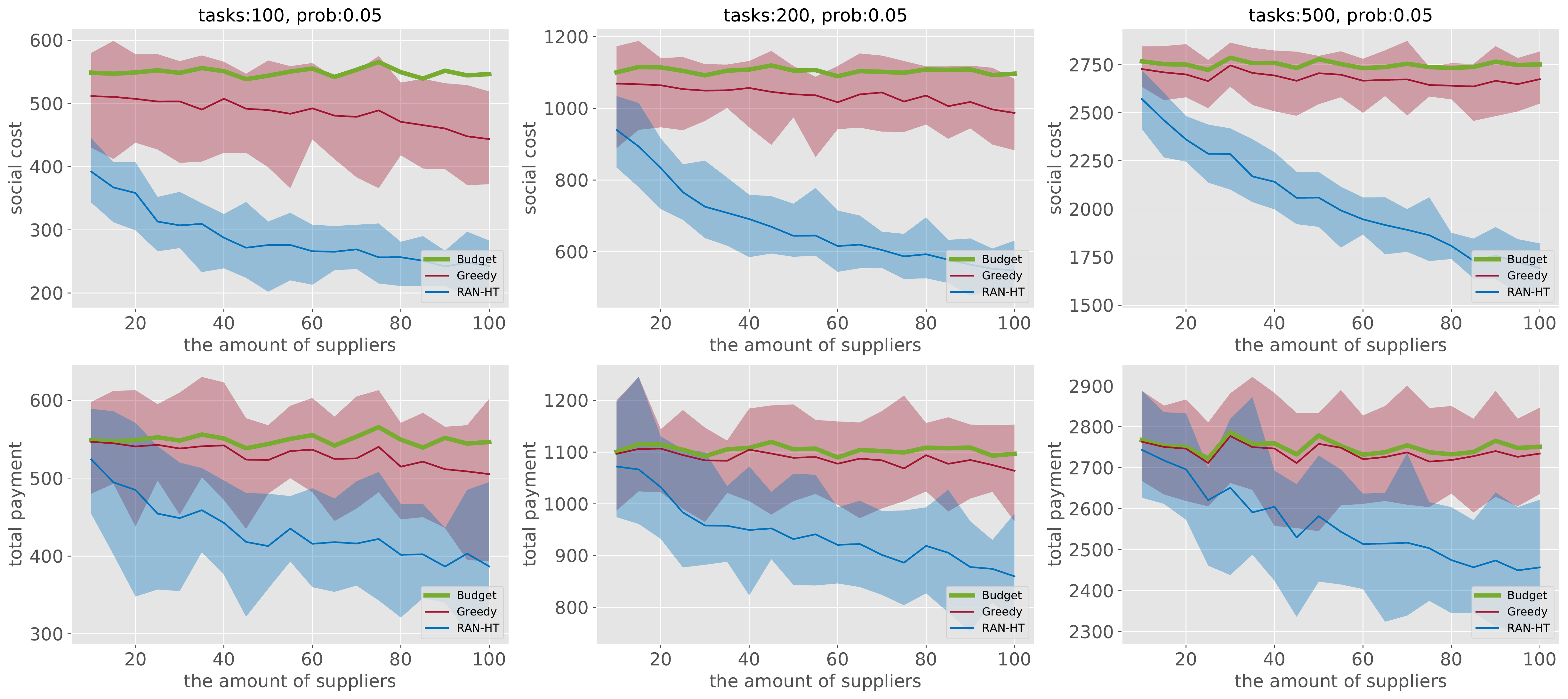}
    \caption{Social Cost and Payment under different amount of suppliers}
    \label{ran-3}
\end{figure}


It is obvious that our RAN-HT mechanism always outperforms the ordinary greedy mechanism. In each row, for both the social cost and the total payment, the gap between the experimental results of the two mechanisms is becoming increasingly apparent with the amount of suppliers in each market gets larger. For the greedy mechanism, since it is conducted in local markets where the requester's neighbors consists of the total suppliers set. In horizontal view, we would see that with the tasks number increase, the social cost and payment for the greedy mechanism are gradually moving closer to the budget. The reason behind it is mainly that although we construct larger market with more suppliers, the change of local market is quiet little, however, with the increase of tasks number, more tasks would be satisfied by the requester herself. In contract, with the market expanding, more potential suppliers with lower cost would be chosen by our RAN-HT, so RAN-HT finally results in better social cost and help the requester save some procurement cost.

\section{Conclusion and Discussion}
In this paper, we study the combinatorial reverse auction in social networks. Firstly, we analyse the importance of monotone allocation and critical bid for designing truthful mechanisms under network scenario, and reveal that the challenge for truthfulness lies in allocation monotonicity. Then, by identifying the critical winning cost, we design two mechanisms RAN-HM and RAN-HT which have monotone allocation and corresponding payment for homogeneous and heterogeneous goods. Our work explains that designing multi-item diffusion auction mechanisms under the guidance of allocation rule with valuation and diffusion monotonicity is crucial. One immediate future work is to design one truthful mechanism with strong incentive for propagating auction information. Moreover, exploring combinatorial auction via social network with other valuation types is also interesting. 

\bibliographystyle{plain}
\bibliography{ref}
\end{document}